\documentclass[a4paper,11pt]{article}
\usepackage[latin1]{inputenc}
\usepackage[T1]{fontenc}

\newcommand{\myauthor}{David Monniaux}
\newcommand{\mytitle}{Applying the Z-transform for the static analysis of
  floating-point numerical filters}
\newcommand{\mykeywords}{Z-transform, linear filters, compositional semantics, numerical computation, floating-point}

\usepackage{amsthm,amsfonts,amsmath}
\usepackage{alltt,url}
\usepackage{graphicx,color}
\usepackage[pdfauthor={\myauthor},pdftitle={\mytitle},pdfkeywords={\mykeywords}]{hyperref}

\title{\mytitle}
\author{\myauthor}

\newcommand{\bbR}{\mathbb{R}}
\newcommand{\bbN}{\mathbb{N}}
\newcommand{\bbC}{\mathbb{C}}
\newcommand{\bbQ}{\mathbb{Q}}
\newcommand{\e}{\mathrm{e}}

\newcommand{\eps}{\varepsilon}
\newcommand{\abs}{\textrm{abs}}
\newcommand{\rel}{\textrm{rel}}

\newcommand{\ratring}[1]{#1[z]_{(z)}}
\newcommand{\ratringr}{\ratring{\bbQ}}

\newcommand{\matrices}{\mathcal{M}}
\newcommand{\minor}{\textrm{minor}}
\newcommand{\Id}[1]{\textrm{Id}_#1}

\newtheorem{thm}{Theorem}
\newtheorem{lem}[thm]{Lemma}
\newtheorem{cor}[thm]{Corollary}

\newcommand{\TM}{$^{\text{TM}}$}

\begin{document}
\maketitle

\section{Introduction}
The static analysis of control/command programs, with a view to proving
the absence of runtime errors, has recently picked up steam, with the
inception of analyzers capable of scaling up to real industrial
programs. In particular, it is nowadays possible to build \emph{sound}
and \emph{precise} static analyzers scaling up to realistic industrial
situations. A static analyzer takes as input a program (source code or
object code) and outputs a series of facts, warnings or other
indications obtained by automatic analysis of that program.

A static analyzer is said to be sound if all the facts
that it derives from a program (say, ``variable \texttt{x} is always
positive'') are always true, regardless of how and on which inputs the
program is run. Sound static analyzers are based on a
\emph{semantics}, that is, a mathematical definition of possible
program executions.

It is well-known that any method for program verification cannot be at
the same time sound (all results produced are truthful), automatic (no
human intervention), complete (true results can always be proved)
and terminating (always produces a
result)
\footnote{The formal version of this result is a classic of recursion
  theory, known as Rice's theorem.}
unless one supposes that
program memory is finite and thus that the system is available to
model-checking techniques. As a result, sound static analyzers
are bound to produce \emph{false alarms} sometimes; that is, warnings
about issues that cannot happen in reality. One thus wants analyzers
that are \emph{precise}, that is, model reality so closely that they seldom
produce false alarms --- but also, one wants analyzers that are
efficient, taking only reasonable amounts of time and memory to
perform an analysis.

One crucial class of errors for control/command systems
is arithmetic overflows --- say,
when converting some value to an integer --- in programs using
floating-point computations.
Such errors have already
proved to be extremely dangerous, having for instance caused the
explosion of the Ariane~5 on its maiden flight \cite{Ariane501}. In
order to prove the absence of such errors, static analyzers such as
Astrée%
\footnote{\url{http://www.astree.ens.fr}}
\cite{BlanchetCousotEtAl02-NJ,BlanchetCousotEtAl_PLDI03} have
to bound all floating-point variables in the program. It is impossible
to do so using simple interval arithmetic; in order to bound the
output of a numerical filter, one has to make the analyzer understand
the stability conditions of the numerical processing implemented in
the application to be analyzed.

In current control/command designs, it is commonplace that the
executable is obtained by compiling C code, or assembly code, itself
obtained by automatic translation from a high-level
specification. This high-level specification is typically given 
in a high-level language such as Simulink%
\footnote{Simulink\TM is a tool for modelling dynamic systems and control
applications, using e.g. networks of numeric filters. The control part
may then be compiled to hardware or software.\\
\url{http://www.mathworks.com/products/simulink/}}
Lustre \cite{LUSTRE} or Scade\TM,%
\footnote{Scade is a commercial product based on LUSTRE.\\
\url{http://www.esterel-technologies.com/products/scade-suite/}}
These languages, in their simplest form, consider programs to be the
software counterpart of a network of electronic circuits (filters,
integrators, rate limiters...) connected by wires; this is actually
how several of these languages represents programs graphically.
Several circuits can be grouped into a compound filter.

One advantage of these high-level languages is that their semantics is
considerably cleaner than those of low-level languages such as C. The
filter and compound filter constructions provide natural
``boundaries'' for blocks of computations that belong together and
probably have some interesting and identifiable properties. It is thus
interesting to be able to analyze these languages in a
\emph{compositional} and \emph{modular} fashion; that is, the analysis
of some block (compound filter) is done independently of that of the
rest of the code, and the result of that analysis may be ``plugged
in'' when analyzing larger programs.

This paper deals with the compositional and modular analysis of
\emph{linear} filters. By this, we mean filters that would be linear
had they been implemented over the real field. Of course, in reality,
these filters are implemented over floating-point numbers and none of
the classical mathematical relationships hold. We nevertheless provide
sound semantics for floating-point computations and sound analysis for
such filters.

\subsection{Digital filtering}
Control/command programs in embedded applications often make use of
linear filters (for instance, low-pass,
high-pass, etc.). The design principles of these filters over the real
numbers are well known; standard basic designs (Butterworth,
Chebyshev, etc.) and standard assembly techniques (parallel, serial)
are taught in the curriculum of signal processing engineers.
Ample literature has been devoted to the design of
digital filters implementing some desirable response, for
implementation in silicon or in software, in fixed-point and in
floating-point.\cite{Jackson_Digital_filters}

However, discrete-time filters are often discussed assuming
computations on real numbers. There is still some considerable
literature on the implications of fixed-point or floating-point
numbers, but the vast majority of the work has focused on ``usual
case'' or ``average case'' bounds ---
it is even argued that worst-case bounds on \emph{ideal} filters on
real numbers are
too pessimistic and not relevant for filter
design \cite[§11.3]{Jackson_Digital_filters}.
The study of the
quantization and roundoff noise generated by fixed-point or
floating-point implementations has mostly been done from a stochastic
point of view, in order to prove average case properties.

For our analysis purposes, we need sound worst-case bounds, and
practical means for obtaining them with reasonable computational
resources. For these reasons, the point of view of the designers of
static analyzers is different from that of the filter designers.

A favorite tool of filter designers is the \emph{Z-transform}
\cite[chapter~3]{Jackson_Digital_filters}, with which the overall ideal
(i.e. implemented over the real numbers) transfer function of a filter
is summarized in a rational function with real coefficients, whose
poles and zeroes determine the frequency response. In this paper, we
shall show how we can use this transform to automatically summarize
networks of linear filters; how this transform allows us to compute
precise bounds on the outcome of the filter, and to statically
summarize complex filters; and how to deal with roundoff errors
arising from floating-point computations.

\subsection{Contributions of the article}
This article gives a sound \emph{abstract semantics} for linear numerical
filters implemented in floating-point or fixed-point arithmetics,
given as the sum of
a linear part (using the Z-transform) and a nonlinear part (given using
affine bounds); this latter part comes from the roundoff noise (and,
possibly, some optional losses of linear precision done for the sake
of the speed of the analysis).
(Sect.~\ref{part:real_compositional} for the ideal, linear part,
\ref{part:float_compositional} for the nonlinear part).

In many occasions, the computed bounds are obtained from the norms
(Sect.~\ref{part:bound_response}) of certain power series.
In Sect.~\ref{part:bounding_norm1}, we give effective methods on the
\emph{real} numbers for bounding such norms.
In Sect.\ref{part:numerical} we explain how to implement some of these
methods efficiently and soundly using integer and floating-point
arithmetics. In Sect.~\ref{part:experiments} we study a few cases.

As with other numerical domains such as those
developed for Astrée, we proceed as
follows: the exact floating-point concrete semantics is
overapproximated
by a mathematically simple semantics on real numbers, which is itself
overapproximated by proved bounds, which are themselves further
overapproximated by an executable semantics (implemented partly in exact
arithmetics, partly using some variant of interval floating-point
computations). This ensures the soundness of the effective
computations.

This paper is an extended version of \cite{Monniaux_CAV05}.

\subsection{Introduction to linear filters and Z-transforms}
\label{part:easy_beginning}
\begin{figure}[tb]
\begin{center}
\input{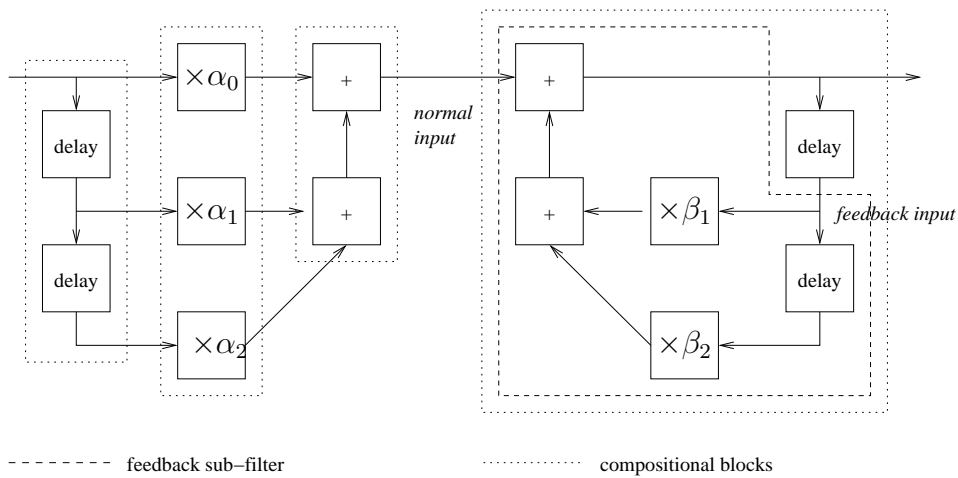}\vspace{-5mm}
\end{center}

\caption{Decomposition of the TF2 filter
  $S_n = \alpha_0 E_n \allowbreak + \allowbreak
         \alpha_1 E_{n-1} \allowbreak + \allowbreak
         \alpha_2 E_{n-2} \allowbreak + \allowbreak
         \beta_1 S_{n-1}  \allowbreak +  \allowbreak
         \beta_2 S_{n-2}$
  into elementary blocks. The compositional blocks are chained by
  serial composition. Inside each compositional on the left,
  elementary gates are composed in parallel. On the right hand side,
  a feedback loop is used.}
\label{fig:composition}
\end{figure}

Let us consider the following piece of C code, which we will use as a
running example (called ``TF2''):
{\small
\begin{alltt}
Y = A0*I + A1*Ibuf[1] + A2*Ibuf[2];
O = Y + B1*Obuf[1] + B2*Obuf[2];
Ibuf[2]=Ibuf[1]; Ibuf[1]=I;
Obuf[2]=Obuf[1]; Obuf[1]=O;
\end{alltt}
}
All variables are assumed to be real numbers
(we shall explain in later sections how
to deal with fixed- and floating-point values with full generality and
soundness). The program takes \texttt{I} as an input and outputs
\texttt{O}; \texttt{A0} etc. are constant coefficients.
This piece of code is wrapped inside a (reactive) loop; the
\emph{time} is the number of iterations of that loop. Equivalently,
this filter can be represented by the block diagram in
Fig.~\ref{fig:composition}.

Let us note $a_0$ etc. the values of the constants and $i_n$
(resp. $y_n$, $o_n$) the value of \texttt{I} (resp. \texttt{Y},
\texttt{O}) at time $n$. 
Then, assuming $o_k=0$ for $k < 0$, we can develop the recurrence:
$o_n
\allowbreak=\allowbreak y_n \allowbreak+\allowbreak b_1.o_{n-1}
                            \allowbreak+\allowbreak b_2.o_{n-2}
\allowbreak=\allowbreak y_n \allowbreak+\allowbreak b_1.
  (y_{n-1} \allowbreak+\allowbreak b_1.o_{n-2}
                            \allowbreak+\allowbreak b_2.o_{n-3})
  \allowbreak+\allowbreak b_2. (y_{n-2} \allowbreak+\allowbreak b_1.o_{n-3}
                            \allowbreak+\allowbreak b_2.o_{n-4})
\allowbreak=\allowbreak y_n \allowbreak+\allowbreak b_1.y_{n-1}
  \allowbreak + \allowbreak (b_2 + b_1^2 b_0) . y_{n-2} + \dots$
where $\dots$ depends solely on $y_k$ with $k < n-2$.
More generally: there exist coefficients $c_0$, $c_1$\dots such that
for all $n$, $o_n = \sum_{k=0} c_k y_{n-k}$. These coefficients solely
depend on the $b_k$; we shall see later some general formulas for
computing them.

But, itself, $y_n = a_0.i_n + a_1.i_{n-1} + a_2.i_{n-2}$. It follows
that there exist coefficients $c'_n$ (depending on the $a_k$ and the
$b_k$) such that  $o_n = \sum_{k=0} c'_k i_{n-k}$. We again find a
similar shape of formula, known as a \emph{convolution product}. The
$c'_k$ sequence is called a \emph{convolution kernel}, mapping $i$ to~$o$.

Let us now suppose that we know a bound $M_I$ on the input:
for all $n$, $|i_n| \leq M_I$; we wish to derive a bound $M_O$ on the
output.
By the triangle inequality, $|O_n| \leq \sum_{k=0} |c'_k| . M_I$. The
quantity $ \sum_{k=0} |c'_k|$ is called the $l1$-norm of the
convolution kernel $c'$.

What our method does is as follows: from the description of a complex
linear filter, it compositionally computes compact, finite representations of
convolution kernels mapping the inputs to the outputs of the
sub-blocks of the filter, and accurately computes the norms of these
kernels (or rather, a close upper bound thereof). As a result, one can
obtain bounds on any variable in the system from a bound on the input.

\section{Linear filters}
In this section, we give a rough outline of what we designate by
linear filters and how their basic properties allow them to be
analyzed.

\subsection{Notion of filters}\label{part:notion_filter}
We deal with numerical filters that take as inputs and output
some (unbounded) discrete streams of floating-point numbers, with
\emph{causality}; that is, the output of the filter at time $t$
depends on the past and present inputs (times 0 to $t$), but not on
the future inputs.%
\footnote{There exist non-causal numerical filtering techniques
One striking example is Matlab's \texttt{filtfilt} function, which
runs the same causal filter in one direction, then in the reverse-time
direction over the same signal; the overall filter has zero phase
shift at all frequencies, a very desirable characteristic in some
applications. Unfortunately, as seen on this example, non-causal
filters require buffering the signal and thus are not usable for
real-time applications. They are outside the scope of this paper.} 
In practice, they are implemented with a  state, and the
output at time $t$ is a function of the input at time $t$ and the
internal state, which is updated.
Such filters are typically implemented as one piece of a synchronous
reactive loop
\cite[§4]{BlanchetCousotEtAl_PLDI03}:
\begin{alltt}
while(true) \{
\ \ ...
\ \ (state, output) = \textbf{filter}(state, input);
\}
\end{alltt}

\subsection{Linear filters}\label{part:linear_filter_formalism}
We are particular interested in filters of
the following form (or compounds thereof):
if $(s_k)$ and $(e_k)$ are respectively the input
and output streams of the filter, there exist real coefficients
$\alpha_0$, $\alpha_1$, \ldots $\alpha_n$ and
$\beta_1$, \ldots $\beta_m$
such that for all time $t$, $s_t$ (the output at time $t$) is defined
as:
\begin{equation}
s_t=\sum_{k=0}^n \alpha_k e_{t-k} + \sum_{k=1}^m \beta_k s_{t-k}
\end{equation}
or, to make apparent the state variables,
\begin{equation}
\begin{bmatrix}
s_{t-m+1}\\
\vdots\\
s_t
\end{bmatrix} =
\begin{bmatrix}
0 & 1 \\
\vdots & \ddots & \ddots \\
0 & \cdots & 0 & 1\\
\beta_m & \cdots & \beta_2 & \beta_1
\end{bmatrix}.
\begin{bmatrix}
s_{t-m}\\
\vdots\\
s_{t-1}
\end{bmatrix} +
\begin{bmatrix}
0 & \cdots & 0\\
\vdots && \vdots\\
0 & \cdots & 0\\
\alpha_n & \cdots & \alpha_0
\end{bmatrix} .
\begin{bmatrix}
e_{t-n}\\
\vdots\\
e_n
\end{bmatrix}
\end{equation}

If the $\beta$ are all null, the filter has necessarily \emph{finite
impulsional response} (FIR) while in the opposite case, it may have
\emph{infinite impulsional response} (IIR). The reason for this
terminology is the study of the reaction of the system to a unit
impulse ($e_0=1$ and $\forall k>0~e_k=0$). In the case of a FIR
filter, $n+1$ time units after the end of the impulse, the output
becomes permanently null. In the case of an IIR filter, the output
(when computed ideally in the real numbers) never becomes permanently
null, but rather follows some exponential decay if the filter is
\emph{stable}. A badly designed IIR filter may be
unstable. Furthermore, it is possible to design filters that should
be stable, assuming the use of real numbers in computation, but that
exhibit gross numerical distortions due to the use of floating-point
numbers in the implementation.

Linear filters are generally noted using their \emph{Z-transform}%
\footnote{An alternate notation \cite{Jackson_Digital_filters}
  replaces all occurrences of $z$ by $z^{-1}$. In such a formalism,
  conditions such as ``the poles must have a module greater
  than~1'' are replaced by the equivalent for the inverse, e.g.
  ``the poles must have a module strictly less than~1''. We chose
  polynomials in $z$ because they allow using normal power series
  instead of Laurent series.}
\begin{equation}
\frac{\alpha_0 + \alpha_1 z + \cdots + \alpha_n z^n}{
            1 - \beta_1 z - \cdots - \beta_m z^m}
\end{equation}
The reasons for this notation  will be made clear in
Sect.~\ref{part:why_z_transform}. In particular, all the ideal compound 
linear filters expressible with elementary elements such as products
by constants, delays, etc. can be summarized by their Z-transform
(Sect.~\ref{part:real_compositional}); that is, they are equivalent to 
a filter whose output is a linear combination of the last $n$ inputs
and $m$ outputs. The Z-transform will also be central in the semantics
of floating-point and fixed-point filters
(Sect.~\ref{part:float_compositional}).

To summarize some salient points of the following sections,
FIR filters given by $\alpha$'s are very easy to
deal with for our purposes, while the stability and decay conditions
of IIR filters are determined by the study of the above rational
function and especially the module of the zeroes of the
$Q(z)= 1 - \beta_1 z - \cdots - \beta_m z^m$ polynomial ($z_0$ is a
\emph{zero} of $Q$ if $Q(z_0)=0$). Those roots are the inverses of the
eigenvalues of the transition matrix.
Specifically, the filter is stable
if all the zeroes have module greater than 1.

\subsection{Bounding the response of the filter}
\label{part:bound_response}
The output streams of a linear filter, as an element of $\bbR^\bbN$,
are linear functions of the inputs and the initial values of the
state variables (internal state variables).

More precisely, we shall see later that, neglecting the floating-point
errors and assuming zero in the initial state variables, the output $S$ is the
\emph{convolution product} $Q \star E$
of the input $E$ by some \emph{convolution
kernel} $Q$: there exists a sequence $(q_n)_{n \in \bbN}$ of reals
such that for any $n$, $s_n = \sum_{k=0}^n q_k e_{n-k}$.
The filter is FIR if this convolution kernel is null except for the
first few values, and IIR otherwise.
If the initial state values $r_1$, \ldots, $r_n$ are nonzero, then
$S = Q_0 \star E + r_1 Q_1 + r_n Q_n$ where the $Q_i$ are convolution
kernels.

Let $E: (e_k)_{n \in \bbN}$ be a sequence of real or complex
numbers. We call \emph{L$\infty$-norm} of $E$, if finite, and note
$\| E \|_\infty$ the quantity $\sup_{k \in \bbN} | e_k |$.
Because of the isomorphism between sequences and formal power series,
we shall likewise note $\| \sum_k a_k z^k \|_\infty = \sup_k |a_k|$.
We are interested in bounding the response of the filter with respect
to the infinite norm: i.e. we want to construct a function $f$ such
that $\| S \|_\infty \leq f(\| E \|_\infty)$.
Said otherwise, if for all the past of the computation since the
last reset of the filter, $|e|$ was less than $M$, then has $|s|$ has
been always less than $f(|M|)$ since the last reset.

If we do not have initialization conditions nor floating-point
errors, $f$ will be linear, otherwise it will be affine.
Let us place ourselves for now in the former case: we are trying to
find a number $g$ such that $\| S \|_\infty \leq g.\| E \|\infty$.
For any linear function $f$ mapping sequences to sequences, we call
\emph{subordinate infinite norm} of $f$, noted, $\| f \|_\infty$
the quantity $\sup_{\|x\|_\infty=1} \| f(x) \|_\infty$, assuming is is
finite. We are thus interested in
$g=\| E \mapsto Q \star E \|_\infty$. If this quantity is finite,
the filter is stable; if it is not, it is
unstable: it is possible to feed an input sequence
to the filter, finitely bounded, which we result in arbitrarily high
outputs at some point in time.

For a sequence (or formal series) $A$, we note
$\| A \|_1 = \sum_{k=0}^\infty |a_k|$, called its \emph{L1-norm}, if
finite. Then we have the following crucial and well-known result
\cite[§11.3]{Jackson_Digital_filters}: % TODO: verifier ref

\begin{lem}
$\| E \mapsto Q \star E \|_\infty = \| Q \|_1$.
\end{lem}

\begin{proof}
We shall first prove that
$\| E \mapsto Q \star E \|_\infty \leq \| Q \|_1$;
that is, for any sequences $Q$ and $E$,
$\| Q \star E \|_\infty \leq \| Q \|_1 . \| E \|_\infty$.
Let us note $C=Q \star E$. $c_n = \sum_{k=0}^n q_k e_{n-k}$,
therefore $|c_n|
\allowbreak\leq\allowbreak \sum_{k=0}^n |q_k| |e_{n-k}|
\allowbreak\leq\allowbreak \|e\|_\infty. \sum_{k=0}^n |q_k|
\allowbreak\leq\allowbreak \|e\|_\infty . \|Q\|_1$.

We shall then show equality.
Let $M < \| O \|_1$. Recall that
$\| Q \|_\infty = \sum_{k=0}^\infty |q_k|$. Then there exists $N$ such
that $\sum_{k=0}^N |q_k| \geq M$.
Choose $e_k=1$ if $k \leq N$ and $q_{n-k} \geq 0$, $e_k=-1$ otherwise.
Clearly, $\| E \|_\infty=1$, and
$c_n = \sum_{k=0}^n e_k q_{n-k} = \sum_{k=0}^n |q_{n-k}| \geq M$,
therefore $\| Q \star E \|_\infty \geq M$ and
$\| E \mapsto Q \star E\|_\infty \geq M$.
Since this is valid for any~$M < \| Q \|_1$, then the
$\| E \mapsto Q \star E\|_\infty = \| Q \|_1$ equality holds.
\end{proof}

Note that most of the discussion on numerical filters found in the
signal processing literated is based on the L2-norm
$\| x \|_2 = \left(\sum_{k=0}^\infty |x_k|^2\right)^{1/2}$ (which is
adapted to energy considerations) --- for instance, for estimating
the frequency spectrum of the rounding noise. We shall
never use this norm in this article.

\section{Convolution kernels as formal power series}
In the preceding section, we said that the output of the ideal filter is
just the convolution of the input with some (possibly infinite)
kernel. In this section, we show how \emph{formal power series} are a
good framework for describing this convolution, and basic facts
about the kernels of
the filters we are interested, given as \emph{rational functions}.

\subsection{Formal power series}
\label{part:formal_power_series}
We shall first recall a few definitions and facts about formal power
series.
The algebra formal power series  $K[[X]]$ over a field $K=\bbR$ or
$\bbC$ is the vector space of countably infinite sequences $K^\bbN$ where the
product of two sequences $A: (a_k)_{k\in\bbN}$ and  $B: (b_k)_{k\in\bbN}$ is
defined as $A.B: (c_k)_{k\in\bbN}$ by, for all $n \in \bbN$,
$c_n=\sum_{k=0}^n a_k b_{n-k}$ (convolution).
Remark that for any algebra operation (addition, subtraction,
multiplication) and any $N$, we obtain the same results for the coefficients
$c_n$ for $n \leq N$ as if $A$ and $B$ were the coefficients of
polynomials and we were computing the coefficient $c_n$, the $n$-th
degree coefficient of the polynomial $A.B$.%
\footnote{One can therefore see $K[[X]]$ as the projective limit of the
$K[X]/X^n$ quotient rings with the canonical
$K[X]/X^{n+1} \rightarrow K[X]/X^n$ morphisms in the category of rings.}
For this reason, we shall from now on note
$A(z)=\sum_{k=0}^\infty a_k z^k$
by analogy with the polynomials. Note that for most of this article, we are
interested in \emph{formal} power series and not with their possible
interpretation as holomorphic functions (i.e. it is not a problem at
all if the convergence radius of the $\sum_{k=0}^\infty a_k z^k$
series is null); we shall note the rare occasions when we need
convergence properties (and we shall prove the needed convergences).
If all the $a_k$ are null except for a finite number,
the formal series $A$ is a polynomial.

Wherever we have a convolution $(a_k) \star (b_k)$ of sequences, we can
equivalently consider a product $A.B$ of formal series.

We shall often wish to take the \emph{inverse} of a power series, and
the quotient $A/B$ of two series. This
is possible for any series $\sum_k b_k b^k$ such that $b_0$ is not
null.
We define a sequence of series $A^{(n)}$ as follows: $A^{(0)}=A$,
$A^{(n+1)}=A^{(n)}-q_n*z^k B$ where $q_n=a^{(n)}_n/b_0$. Note that
for all $n \in \bbN$,
$k < n~ A^{(n)}_k=0$ and
$A=A^{(n+1)}+\sum_{k=0}^n q_k z^k B$; thus for all $n$,
$A \equiv \sum_{k=0}^n q_k z^k B \pmod{X^n}$, which may equivalently
written as   $A \equiv Q.B \pmod{X^n}$. Therefore, $A = Q.B$, which
explains why $Q$ can be called the \emph{quotient} of $A$ by $B$.

A very important case for the rest of the paper
is $1/(1-z)=\sum_{k=0}^\infty z^k$.
Another important constatation is
that this quotient formula applied to
\begin{equation}
S = E.\frac{\alpha_0 + \alpha_1 z + \cdots + \alpha_n z^n}{
            1 - \beta_1 z - \cdots - \beta_m z^m}
\end{equation}
where $S$ and $E$ are expressed as formal power series
is equivalent to running the IIR filter defined by the above rational
function with $E$ the inputs and $S$ the output.

\subsection{Stability condition}
We manipulate convolution kernels expressed as
rational functions where the coefficient of degree $0$ of the
denominator is $1$. We shall identify a rational function with the
associated formal power series.
Using complex analysis, we shall now prove the following lemma, giving
the stability condition familiar to filter designers:

\begin{lem}
$\| Q \|_1 < \infty$ if and only if all the poles of $Q$ are
outside of the $|z| \leq 1$ unit disc.
\end{lem}

That is: a filter is stable in ideal real arithmetics if and only if
all its poles have module greater than 1.

\begin{proof}
Consider the poles of the rational function $Q$. If none are in
the $|z| \leq 1$ unit disc, then the radius of convergence of the
power series of the meromorphic function $Q$ around 0 has a radius of
convergence strictly greater than 1. This
implies that the series converges absolutely for $z=1$ and thus that
$\| Q \|_1$ is finite. % TODO: verifier terminologie
On the other hand, if $\| Q \|_1 < \infty$ then the series converges
absolutely within the $|z| \leq 1$ unit disc and no pole can be within
that disc.
\end{proof}

\section{Compositional semantics: real field}
\label{part:real_compositional}
Now, we have a second look at the basic semantics of linear filters,
in order to give a precise and
 compositional \emph{exact} semantics of
compound filters on the real numbers. We show
that any linear filter with one input and one output
is equivalent (on the real numbers) to a filter
as defined in §\ref{part:linear_filter_formalism}.

\subsection{Definition}
A filter or filter element has
\begin{itemize}
\item $n_i$ inputs $I_1$, \dots, $I_{n_i}$ (collectively, vector $I$),
  each of which is a \emph{stream} of real numbers;
\item $n_r$ reset state values $r_1$, \dots, $r_{n_r}$ (collectively,
  vector $R$), which are the initial values of the state of the
  internal state variables of the filter (inside delay operators)
  at the beginning of the computation;
\item $n_o$ output streams $O_1$, \dots, $O_{n_o}$ (collectively,
  vector $O$).
\end{itemize}

If $M$ is a matrix (resp. vector) of rational functions, or series,
let $N_x(M)$ be the coordinate-wise application of the norm $\| \cdot \|_x$ to
each rational function, or series, thereby providing a vector
(resp. matrix) of nonnegative reals. We note $m_{i,j}$ the element in
$M$ at line $i$ and column $j$.

We note by $\bbR(z)$ the field of rational functions over $\bbR$ and
by $\ratringr$ the ring of rational functions of the form
$P(z)/(1-zQ(z))$ where $P$ and $Q$ are polynomials (that is, the ring
of rational functions such that the constant term of the denominator
is not null).%
\footnote{This last ring is the localization of the ring $\bbR[z]$ of
  real polynomials at the prime ideal $(z)$ generated by $z$, thus the
  notation.}
When $F \in \ratringr$, we note $\| F \|_1$ the
L1-norm of the associated power series.

When computed upon the real field, a filter $F$ is characterized by:
\begin{itemize}
\item a matrix $T^F \in \matrices_{n_o,n_i}(\ratringr)$ such that
  $t_{i,j}$ characterizes the linear response of output stream $i$
  with respect to input stream~$j$;

\item a matrix $D^F \in \matrices_{n_o,n_r}(\ratringr)$ such that
  $d_{i,j}$ characterizes the (decaying) linear response of output stream $i$
  with respect to reset value~$j$.
\end{itemize}
We note $F(I,R)$ the vector of output streams of filter $F$ over the
reals, on the vector of input streams $I$ and the vector of reset values~$R$.
Then we have
\begin{equation}
\forall I \in (\bbR^\bbN)^{n_i}~
\forall R \in \bbR^{n_r}~
 F(I,R) = T^F.I + D^F.R
\end{equation}

When the number of inputs and outputs is one, and initial values are
assumed to be zero, the characterization of the filter is much simpler
--- all matrices and vectors are scalars (reals, formal power series
or rational functions), and
$T^D$ is null. We recommend that the reader
instantiates our framework on this case for better initial
understanding.

\begin{figure}
\begin{center}
\input{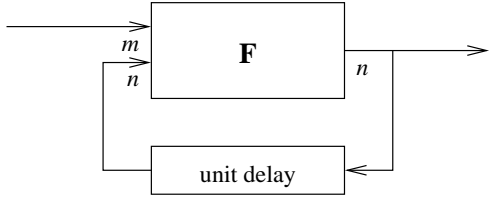}
\end{center}
\caption{A feedback filter}
\label{fig:feedback}
\end{figure}

\subsection{Basic arithmetic blocks}
\begin{description}
\item[Plus] node implemented in floating point type $f$:
 $n_i=n_o=1$,\\
  $T=\begin{bmatrix}1 & 1\end{bmatrix}$, $D=[]$;
\item[Scale by $k$] node implemented in floating point
  type $f$: 
  $T=\begin{bmatrix}k\end{bmatrix}$, $D=[]$;
\item[Delay without initializer] (delay for $n$ clock ticks):
  $T=\begin{bmatrix}z^n\end{bmatrix}$,
  $D=0$;
\item[Unit delay with initializer]:
  $T=\begin{bmatrix}z\end{bmatrix}$,
  $D=\begin{bmatrix}1\end{bmatrix}$;
\end{description}

\subsection{Composition}
\begin{description}
\item[Parallel composition]
  $T=\begin{bmatrix}T_1 & 0 \\ 0 & T_2\end{bmatrix}$,
  $D=\begin{bmatrix}D_1 & 0 \\ 0 & D_2\end{bmatrix}$;

\item[Serial composition] through filter 1, then 2:\\
  $T=T_2.T_1$,
  $D=\begin{bmatrix} T_2.D_1 & D_2\end{bmatrix}$.
\end{description}

\subsection{Feedback loops}
Let us consider a filter consisting of a filter $F$ with $m+n$ inputs
and $n$ outputs and feedback loops running the $n$ outputs to the last
$n$ inputs through unit delays. (Fig.~\ref{fig:feedback})
We split $T^F$ into sub-matrices
$T_I \in \matrices_{n,m}(\ratringr)$ and
$T_O \in \matrices_{n,n}(\ratringr)$ representing respectively the
responses to the global inputs and to the feedback loop.
The system then verifies the linear equation over the vectors of
formal power series:
$O = T^F_I.I + z T^F_O.P + D.R$, and thus
$(\Id{n}-z T^F_O) O = T^F_I.I + D^F.R$.

By Cor.~\ref{cor:feedback_matrix_inversion},
$\Id{n}-z T^F_O$ is invertible in $\matrices_{n,n}(\ratringr)$,%
\footnote{
This result is not surprising, because the system, by construction,
must admit causal solutions.}
thus $T = (\Id{n}-z T^F_O)^{-1}.T^F_I$
and $D = (\Id{n}-z T^F_O)^{-1}.D^F$.
Section~\ref{part:approximate_algebraic_structures} explains how to
perform such computations in practice.

\subsection{Examples}
\label{part:why_z_transform}
\label{part:reconstitution}

\begin{figure}
\begin{center}
\input{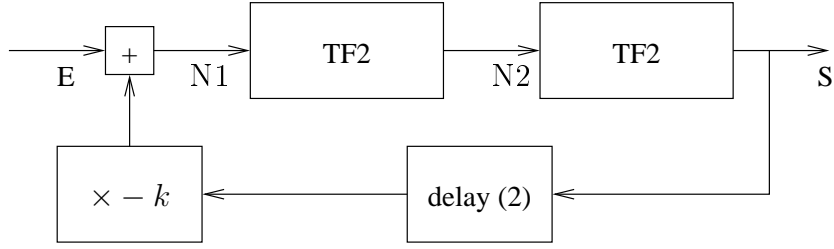}
\end{center}
\caption{A compound filter consisting of two second order filters and
  a feedback loop. Each TF2 node is a second-order filter whose
  transfer function is of the form
$(\alpha_0+\alpha_1 z+\alpha_2 z^2)(1-\beta_1 z-\beta_2 z^2)^{-1}$.}
\label{fig:compound_filter_example1}
\end{figure}

A second order IIR linear filter is expressed by
$S = \alpha_0.E + \alpha_1.\textrm{delay}_2(E) +
\alpha_2.\textrm{delay}_2(E) + \beta_1.\textrm{delay}_1(S)
+ \beta_2.\textrm{delay}_2(S)$. This yields an equation
$S = (\alpha_0+\alpha_1 z+\alpha_2 z^2)E + (\beta_1 z+\beta_2 z^2)S$.
This equation is easily solved into
$S=(\alpha_0+\alpha_1 z+\alpha_2
  z^2)(1-\beta_1 z-\beta_2 z^2)^{-1}.E$.

In Fig.~\ref{fig:compound_filter_example1}, we first analyze the two
internal second order IIR filters separately and obtain
\begin{eqnarray}
Q_1=\frac{\alpha_0+\alpha_1 z+\alpha_2 z^2}{1-\beta_1 z-\beta_2 z^2}\\
Q_2 = \frac{a_0+a_1 z+a_2 z^2}{1-b_1 z-b_2 z^2}\\
\end{eqnarray}

The we analyze the feedback loop and obtain for the whole filter a
rational function with a 6th degree dominator:
\begin{equation}
S = \frac{Q_1.Q_2}{1+kz^2.Q_1.Q_2}.E
\end{equation}
where $Q_1$ and $Q_2$ are the transfer function of the TF2 filters
(form $(\alpha_0+\alpha_1 z+\alpha_2 z^2)(1-\beta_1 z-\beta_2
z^2)^{-1}$), which we computed earlier.

\section{Bounding the 1-norm of series expansions of rational
  functions}
\label{part:bounding_norm1}

\subsection{Inverses of products of affine forms}
\label{part:inverse_product_norm}
Let $\xi_i$ be complex numbers of module strictly greater than 1.
Let $Q(z)$ be the formal power series $\prod_{i=1}^n Q_i$
where the $Q_i(z)$ are the power series $(z-\xi_i)^{-1}$.
The $n$-th degree coefficient of $q_i$ is $-\xi_i^{(n+1)}$, by the
easy expansion:
\begin{equation}
\frac{1}{z-\xi} = \frac{-1/\xi}{1-z/\xi}
\end{equation}
$q^{(n)}$, the coefficient of $z^n$ in the $Q$ power series, is
obtained by successive convolution products; it is
\begin{equation}
q^{(n)}=\sum_{\forall i, k_i \in \bbN \wedge \sum_i k_i=n}
  \prod q_i^{(k_i)}
\end{equation}
We can therefore bound its module:
\begin{equation}
\left|q^{(n)}\right| \leq \sum_{\forall i, k_i \in \bbN \wedge \sum_i k_i=n}
  \prod \left|q_i^{(k_i)}\right|
\end{equation}

The right hand side of the preceding inequality is just the
coefficient $\dot{q}^{(n)}$ of the series
$\prod_{i=1}^n \dot{Q}_i$ where
$\dot{q}_i^{(n)}=\left|q_i^{(n)}\right|=|\xi_i|^{(n+1)}$ is the $n$-th order
coefficient of the $\frac{1}{|\xi_i|-z}$ series.
Since $|\xi_1| > 1$,
the convergence radius of this last series
is strictly greater than 1; furthermore, all its coefficients are
nonnegative; therefore, the sum
of its  coefficients is the value of the function at $z=1$,
that is, $\frac{1}{|\xi_i|-1}$.
We can therefore give an upper bound:
\begin{equation}
\left\| \frac{1}{(z-\xi_1)\cdots(z-\xi_n)} \right\|_1
\leq \frac{1}{(|\xi_1|-1)\cdots(|\xi_n|-1)}
\end{equation}
%% and, more precisely, for all $N \in \bbN$:
%% \begin{equation}
%% \left\| \frac{1}{(z-\xi_1)\cdots(z-\xi_n)} \right\|_1 ^{\geq N}
%% \leq \frac{1}{\prod_{i=1}^n (|\xi_i|-1)} -
%%   \left\| \frac{1}{\prod_{i=1}^n (z-|\xi_i|)} \right\|_1 ^{< N}
%% \end{equation}
%% In section~\ref{part:compute_initial_norm1} we shall discuss the
%% numerical aspects of this computation.

\subsection{Rough and less rough approximation in the general case}
\label{part:norm1_rough_approximation}
Let $P(z)/Q(z)$ be a rational function, with
$P(z)$ a polynomial of degree $m$
$Q(z)$ a monic polynomial of degree $n$.
Let $\textrm{zeroes}(Q)$ be the multiset of zeroes of
$Q$ (multiple zeroes are counted with their multiplicity).
$P(z)=\sum_k p_k z^k Q(z)$, thus $\|P\|_1 \leq \sum_k |p_k|.\|Q\|_1$.
Therefore
\begin{equation}
\left\| \frac{P}{Q} \right\|_1 \leq
  \frac{\| P \|_1}{\prod_{\xi \in \textrm{zeroes}(Q)}(|\xi|-1)}
\end{equation}

This is, however, a very coarse approximation. Intuitively, the mass
of the convolution kernel expressed by the $P/Q$ series lies in its
initial terms. Still, with the above formula, we totally neglect the
cancellations that happen in the computation of this initial part of
the kernel; i.e. instead of considering $|a-b|$, we bound it by
$|a|+|b|$. The solution is to split $\| P/Q \|_1$ into
$\| P/Q \|_1^{<N}$ and $\| P/Q \|_1^{\geq N}$.
We shall elaborate on this in Sect.~\ref{part:development}.

\subsection{Second degree denominators with complex poles}
A common case for filtering applications is when the denominator is a
second degree polynomial $Q$ of negative discriminant. In this case,
the roots of $Q$ are two conjugate complex numbers $\xi$ and
$\bar{\xi}$ and the decomposition is as follows:
\begin{equation}
\frac{P(z)}{Q(z)} = P_0(z) +
  \frac{\lambda}{z-\xi} + \frac{\bar{\lambda}}{z-\bar{\xi}}
\end{equation}
where $\lambda=P(\xi)/(\xi-\bar{\xi})$. We shall for now leave $P_0$ out.

We are interested in the coefficients $a_k$ of this series:
\begin{equation}
a_k = -\left(\frac{\lambda}{\xi^{k+1}}+
             \frac{\bar{\lambda}}{\bar{\xi}^{k+1}}\right)
\end{equation}
Let us write $\lambda=|\lambda| \e^{i\alpha}$ and
$\xi=|\xi| \e^{i\beta}$; then
\begin{multline}
a_k = - \frac{|\lambda|}{|\xi|^{k+1}}
   \left(\e^{i\alpha}.\e^{-i(k+1)\beta}
   + \e^{-i\alpha}.\e^{i(k+1)\beta}\right)\\
   = - 2 \frac{|\lambda|}{|\xi|^{k+1}} \cos\left(\alpha-(k+1)\beta\right) 
\end{multline}

To summarize, the sequence is a decreasing exponential of rate $1/|\xi|$
modulated by a sine wave and multiplied by a constant factor
$|\lambda|/|\xi|$. Therefore, computing $|\lambda|$ and $|\xi|$ will
be of prime importance.
If $Q$ is monic $Q(z)=z^2 + z_1 x + z_0$, then
$|\xi|^2 = \xi\bar{\xi} = c_0$. In the case of a rational function
of the form
\begin{equation}
\frac{P(z)}{Q(z)} =
  \frac{\alpha_0+\alpha_1 z+\alpha_2 z^2}{1-\beta_1 z-\beta_2 z^2}
\end{equation}
then $|\xi| = |\beta_2|^{-1/2}$ and $\lambda=P(\xi)/(\xi-\bar{\xi})$.
Should we prefer not to compute with complex numbers,
\begin{equation}
|\lambda|^2=\lambda\bar{\lambda} =
  \frac{P(\xi)(\bar{\xi}-\xi)+P(\bar{\xi})(\xi-\bar{\xi})}{
       (\xi-\bar{\xi})^2}
\end{equation}
The numerator is a symmetric polynomial in $\xi$ and $\bar{\xi}$,
roots of $Q$, and therefore can be expressed as a polynomial in the
coefficients of $Q$~; its coefficients are polynomials in the
coefficients of $P$, therefore the whole polynomial can be expressed
as a polynomial in the coefficients of $P$ and $Q$. The denominator is
just the discriminant of $Q$.
\begin{equation}\small
|\lambda|^2 = 
\frac{{\alpha_2}^2 + 
      \beta_2\,\left( -{\alpha_1}^2 - 
         \alpha_0\,\alpha_1\,\beta_1 + 
         {\alpha_0}^2\,\beta_2 \right)  + 
      \alpha_2\,\left( \alpha_1\,
          \beta_1 + 
         \alpha_0\,
          \left( {\beta_1}^2 + 2\,\beta_2
            \right)  \right) }{-({\beta_1}^2 + 
      4\,\beta_2)}
\end{equation}

We are now interested in bounding $|a_k|$. If we just use
$| \cos\left(\alpha \allowbreak - \allowbreak
(k+1)\beta\right) | \allowbreak \leq \allowbreak
1$,
we come back to the
earlier bounds obtained by totally separating the series arising from
the two poles.

We shall now obtain a better bound using the following constatation:
for any real $\theta$,
\begin{equation}
|\cos \theta| = \sqrt{\cos^2 \theta}
              = \sqrt{\left(1+\cos(2\theta)\right)/2}
\leq 2^{-1/2} (1+\cos(2\theta)/2)
\end{equation}
using the concavity inequality $\sqrt{1+x} \leq 1+x/2$.
Therefore
\begin{equation}
|a_k| \leq \sqrt{2} \frac{|\lambda|}{|\xi|^{k+1}}
  (1+\cos(2(\alpha-(k+1)\beta))/2)
\end{equation}
Now, we are interested in bounding $\sum_{k=N}^\infty [a_k|$.
For any $a$ and $b$, and $0 \leq r < 1$
\begin{equation}
\sum_{k=0}^\infty \cos(a+kb) r^k =
  \frac{\cos a - r \cos(a-b)}{1-2r\cos b + r^2}
\end{equation}

Let us now see the quality of such bounds $S_1 \leq S_2$, $S_1 \leq S_3$:
\begin{eqnarray}
S_1 = \sum_{k=0}^\infty |\cos(a+kb)| r^k\\
S_2 = \sum_{k=0}^\infty r^k = \frac{1}{1-r}\\
S_3 = \frac{1}{\sqrt{2}} \sum_{k=0}^\infty (1+2\cos(2(a+kb)))r^k\\
S_3 = \frac{1}{\sqrt{2}} \left(\frac{1}{1-r}
  + \frac{1}{2}. \frac{\cos(2a) - r \cos(2(a-b))}{1-2r\cos b + r^2}\right)
\end{eqnarray}
Note that $S_3$ is not necessarily better than $S_2$ (for $a=0$ and $b=0$,
$S_3/S_2 = 3/(2\sqrt{2}) \simeq 1.06$). However, some moderate gains
may be obtained\,; for instance, for $r=0.7$, $a=0$ and $b=0.3$,
$S_1 \simeq 2.60$, $S_2 \simeq 3.33$ and $S_3 \simeq 2.80$.
For practical purposes, the bound obtained using $S_2$ is very
sufficient and easy to compute. We thus opt for this one.

\subsection{Finer bounds using partial fraction decomposition}
\label{part:partial_fraction_decomposition}
It is well known that if $Q_i$ are pairwise prime polynomials, and $Q$
is their product, then for any polynomial $P$ prime with $Q$ the
fraction $P/Q$ admits a \emph{partial decomposition} as
$P/Q=P_0 + \sum_i P_i/Q_i$, where $P_0$ is the Euclidean quotient of
$P$ by $Q$ and the degree of $P_i$ is strictly less than that of
$Q_i$.

Using the fundamental theorem of algebra, it follows that if the
$\xi_i$ are the distinct roots of $Q$ and $m_i$ their multiplicity,
then there exist $\lambda_{i,j} \in \bbC$ such that
\begin{equation}
P/Q = P_0 + \sum_i \sum_{j=1}^{m_i} \frac{\lambda_{i,j}}{(z-\xi_i)^j}
\end{equation}
Since $Q$ is a real polynomial, its roots are either real, either
pairs of $\xi_i$ and conjugate $\xi_{i'}=\bar{\xi}$, with the same
multiplicity, and also for all $j$,
$\lambda_{i',j}=\bar{\lambda}{i',j}$.

However, while theoretically sound, this result is numerically
delicate when there are multiple roots, or different roots very close
to each other.\cite[§1.3]{MR1646107}
For instance, let us consider a first-degree polynomial
$P$ and a second-degree polynomial $Q$, then
\begin{equation}
\frac{P(z)}{Q(z)}=\frac{\lambda_1}{z-\xi_1}+\frac{\lambda_2}{z-\xi_2}
\end{equation}
and we obtain $\lambda_1=P(\xi_1)/(\xi_2-\xi_1)$ (and
$\lambda_2=P(\xi_2)/(\xi_1-\xi_2)$). Both numbers will get very large,
in inverse proportion of $\xi_1-\xi_2$. While it is quite improbable
that we should analyze filters where two separate poles have been
intentionally be placed very close together, it is possible that we
analyze filters with multiple poles (for instance, the composition of
a filter with itself), and, with numerical computations, we would have
two extremely close poles and thus a dramatic numerical instability.

We still can proceed with a \emph{radius $r$ decomposition} of $P/Q$
\cite[Def~1.3]{MR1646107}: instead of factoring $Q$ into a product of
$z-\xi_i$ factors, we factor it into a product of $Q_i$ such that for
any $i$, and any roots $\xi_1$ and $\xi_2$ of $Q_i$, then
$|\xi_1 - \xi_2| < 2r$. The same reference describes algorithms for
performing such decompositions. We obtain a decomposition of the form
\begin{equation}
\frac{P}{Q} = P_0 + \sum_i P_i/Q_i
\end{equation}
where the roots of each $Q_i$ are close together, the degree of $P_i$
is less than the degree of $P_i$. From this we obtain the bound
\begin{equation}
\left\| \frac{P}{Q} \right\|_1 \leq \|P_0\|_1 +
  \sum_i \| P_i \|. \left\| \frac{1}{Q_i} \right\|_1
\end{equation}
which we can bound using the inequalities given in the preceding
subsections. We can, as before, improve on this bound by splitting the
series between an initial sequence and a tail.
%% \begin{equation}
%% \left\| \frac{P}{Q} \right\|_1 \leq
%%   \left\| \frac{P}{Q} \right\|_1^{<N} +
%%   \sum_i \|P_i\|_1. \left\|\frac{1}{Q_i}\right\|_1^{<N-\textrm{degree}(P_i)}
%% \end{equation}

\subsection{Development of rational functions and normed bounds}
\label{part:development}
Let $P(z)/Q(z) \in \ratringr$ be a rational function representing a
power series by its development $(u_n)_{n\in\bbN}$ around 0. We wish
to bound $\| u \|_1$, which we shall note $\| P/Q \|_1$.
As we said before, most of the mass of the development of $P/Q$ lies
in its initial terms, whereas the ``tail'' of the series is negligible
(but must be accounted for for reasons of soundness). We thus split
$P/Q$ into an initial development of $N$ terms and a tail, and use
\begin{equation}
\| P/Q \|_1 = \| P/Q \|_1^{<N} + \| P/Q \|_1^{\geq N}
\end{equation}
$\| P/Q \|_1$ is computed by computing explicitly the $N$ first terms
of the development of $P/Q$. We shall see in
Sect.~\ref{part:compute_initial_norm1} the difficulties involved in
performing such a computation soundly using interval arithmetics.

Let $d_Q$ be the degree of $Q$. The development $D$ of $P/Q$ yields
an equation $P(z) = D(z).Q(z) + R(z).z^N$.
We have $P(z)/Q(z) = D(z) + R(z)/Q(z).z^N$ and thus
\begin{equation}
\| P/Q \|_1^{\geq N} = \| R/Q \|_1 \leq \| R \|_\infty . \| 1/Q \|_1
\end{equation}
The preceding sub-sections give a variety of methods for bounding
$\| 1/Q \|_1$ using the zeroes of $Q(z)$;
Section~\ref{part:norm1_rough_approximation} gives a rough method based on
lower bounds on the absolute values of the zeroes of $Q(z)$.
$\| R \|_\infty$ is bounded by explicit computation of $R$ using
interval arithmetics; as we shall see, we compute $D$ until the sign
of the terms is unknown --- that is, when the norm of the developed
signal is on the same order of magnitude as the numerical error on it,
which happens, experimentally, when the terms are small in absolute
values. Therefore, $\| R \|_\infty$ is small, and thus
the roughness of the approximation used $\| 1/Q \|_1$
does not matter much in practice.

%% In Section~\ref{part:compute_initial_norm1}, we shall explain how to
%% bound $\| P/Q \|_1^{<N}$ soundly and precisely. 
%% We also have
%% \begin{equation}
%% \| P/Q \|_1^{\geq N} \leq \sum_{k=0}^m |p_k|.\| 1/Q \|_1^{\geq N-k}
%%   \leq \| P \|_1 . \| 1/Q \|_1^{\geq N-m}
%% \end{equation}
%% and we explained in Sect.~\ref{part:inverse_product_norm} how to bound
%% $\| 1/Q \|_1^{\geq N-m}$.

\section{Precision properties of fixed- or floating-point operations}
\label{part:floating_point_error}
In this section, we shall recall a few facts on the errors introduced
by fixed- and floating-point arithmetics. They will be sufficient for
all our reasonings, without need for further knowledge about numerical
arithmetics.

Most types of numerical arithmetics, including
the widely used IEEE-754 floating-point arithmetic, implemented in
hardware in all current microcomputers, define the result of
elementary operations as follows: if $f$ is the ideal operation
(addition, subtraction, multiplication, division etc.) over
the real numbers and $\tilde{f}$
is the corresponding floating-point operation, then
$\tilde{f} = r \circ f$ where $r$ is a \emph{roundoff} function.
The roundoff function chooses a value $r(x)$ that can be exactly
represented in the used fixed- or floating-point data type, and is
very close to $x$; specifically, most systems, including all IEEE-754
systems, provide the following roundoff functions:%
\footnote{On Intel x86 systems, the description of the exact
properties of the floating-point arithmetics is complicated by the
fact that, by default, with most operating systems and languages, the
80287-compatible floating-point unit performs computations internally
using 80-bit long double precision numbers, even when the compiled
program
suggests the use of standard 64-bit double precision IEEE
numbers. Note that such usage of supplemental precision for
intermediate computations is allowed by the C~standard, for example. %TODO
The final result of the computation may therefore depend on
the register scheduling and optimizations performed by the
compiler.
%% Designers concerned about this may use library functions
%% such as \texttt{setfeprec()}, assembler instructions changing the
%% internal precision, or compiler options (such as \texttt{gcc}'s
%% \texttt{-ffloat-store}) that enforce storing all intermediate
%% quantities in memory.
Since we reason by \emph{maximal errors}, our
bounds are always sound (albeit pessimistic) in the face of such
complications, whatever the compiler and the system do.}
\begin{itemize}
\item round to $0$: $r(x)$ is the representable real nearest to $x$ in
  the direction of $0$;
\item round to $+\infty$: $r(x)$ is the representable real nearest to $x$ in
  the direction of $+\infty$;
\item round to $-\infty$: $r(x)$ is the representable real nearest to $x$ in
  the direction of $-\infty$;
\item round to nearest (generally, the default mode): $r(x)$ is the
  representable real nearest to~$x$.
\end{itemize}

In this description, we leave out the possible generation of special
values such as infinities ($+\infty$ and $-\infty$) and
\emph{not-a-number} (NaN), the latter indicating undefined results
such as $0/0$. We assume as a precondition to the numerical filters
that we analyze that they are not fed infinities or NaNs --- indeed,
in some DSP (digital signal processor) implementations, the hardware is
incapable of generating or using such values, and in many other
implementations the system is configured so that the generation or
usage of infinities issues an exception resulting in bringing
the system into a failure mode. Our framework provides constructive
methods for bounding \emph{any} floating-point quantity $x$ inside the
filters as $\|x\|_\infty \leq c_0 + \sum_{k=1}^n c_k.\|e_k\|_\infty$ where the
$e_k$ are the input streams of the system; it is quite easy to check
that the system does not overflow ($\|x\| < M$); one can even easily
provide some very wide sufficient conditions on the input
($\| e_k \|_\infty \leq (M-c_0) / (\sum_{k=1}^n c_k)$). We will not
include such conditions in our description, for the sake of simplicity.

For any arithmetic operation, the discrepancy between the ideal result
$x$ and the floating-point result $\tilde{x}$ is bounded, in absolute value, by
$\max(\eps_{\textrm{rel}} |x|,\eps_{\textrm{abs}})$
where $\eps_{\textrm{abs}}$ is the \emph{absolute error}
(the least positive floating-point
number)%
\footnote{The absolute error results from the \emph{underflow}
condition: a number close to 0 is rounded to 0. Contrary to overflow
(which generates infinities, or is configured to issue an exception),
underflow is generally a benign condition. However, it precludes
merely relying on relative error bounds if one wants to be sound.}
and $\eps_{\textrm{rel}}$ is the \emph{relative error}
incurred
($\eps_{\textrm{abs}}=2^{-1074} \simeq 4.94\cdot 10^{-324}$
and $\eps_{\textrm{rel}}=2^{-53} \simeq 1.11\cdot 10^{-16}$
for IEEE double precision operations, for the worst case with respect
to rounding modes).
We actually take the coarser inequality
\begin{equation}
|x - \tilde{x}| \leq \eps_{\textrm{rel}}|x|+\eps_{\textrm{abs}}
\end{equation}
See \cite{IEEE754} for more details on
floating-point numbers and \cite{Mine:ESOP04} for more about the
affine bound on the error.

In the case of fixed-point arithmetics, we have $\eps_\rel=0$ and
$\eps_\abs=\delta$ ($\delta$ is the smallest positive fixed-point
number) if the rounding mode is unknown (round to $+\infty$,
$-\infty$ etc.) and $\delta/2$ is it is the rounding mode is known
to be round-to-nearest.

\section{Compositional semantics: fixed- and floating-point}
\label{part:float_compositional}
In this section, we give and a compositional
\emph{abstract} semantics of filters on the floating-point numbers.

\subsection{Constraint on the errors}
Our abstract semantics characterizes a fixed- or floating-point filter
$\tilde{F}$ by:
\begin{itemize}
\item the exact semantics of the associated filter $F$ over the real
 numbers
\item an abstraction of the discrepancy $\Delta(I)=\tilde{F}(I)-F(I)$
  between the ideal and floating-point filters.
\end{itemize}

We transform $\tilde{F}(I)$ into the sum of a term that
we can bound very accurately using algebra and complex analysis, and a
nondeterministic input $\Delta(I)$ that we cannot analyze accurately
and soundly without considerable difficulties, but for which bounds
are available: assuming for the sake of simplicity
a single input and a single output and no initialization conditions,
we obtain an \emph{affine}, \emph{almost linear}
constraint on the $\| \Delta(I) \|\infty$ with respect to
$\| I \|_\infty$:
$\| \Delta(I) \|_\infty \leq \eps^F_\rel \| I \|_\infty + \eps^F_\abs$.
In short: since the filter is linear, the magnitude of the error is
(almost) linear.

We generalize this idea to the case of multiple inputs and outputs.
The abstract semantics characterizing $\Delta$ is given by
matrices $\eps_{\rel,T}^F \in \matrices_{n_o,n_i}(\bbR_+)$
and $\eps_{\rel,D}^F \in \matrices_{n_o,n_r}(\bbR_+)$
and a vector $\eps_\abs^F \in \bbR_+^{n_o}$
such that
\begin{equation}
\| F(I,R) - \tilde{F}(I,R) \|_\infty \leq
  \eps_{\rel,T}^F. N_\infty(I) + \eps_{\rel,D}^F. N_\infty(R) + \eps_\abs.
\end{equation}
where $\tilde{F}(I,R)$ is the output on the stream computed upon the
\emph{floating-point} numbers on input streams $I$ and initial values~$I$.

\subsection{Basic arithmetic blocks}
\begin{description}
\item[Plus] node implemented in floating point type $f$:
 $n_i=n_o=1$,
  $T=\begin{bmatrix}1 & 1\end{bmatrix}$, $D=0$,
  $\eps_{\rel,T}=\begin{bmatrix}\eps_\rel^f &\eps_\rel^f\end{bmatrix}$,
  $\eps_{\rel,D}=0$,
  $\eps_\abs = \eps_\abs^f$;
\item[Scale by $k$] node implemented in floating point
  type $f$: 
  $T=\begin{bmatrix}k\end{bmatrix}$, $D=0$,
  $\eps_{\rel,T} = |k|.\eps_\rel^f$,
  $\eps_{\rel,D}=0$;
  $\eps_\abs = \eps_\abs^f$;
\item[Delay without initializer] (delay for $n$ clock ticks):
  $T=\begin{bmatrix}z^n\end{bmatrix}$,
  $D=0$,
  $\eps_{\rel,T}=0$, $\eps_{\rel,D}=0$, $\eps_\abs=0$
\item[Unit delay with initializer]:
  $T=\begin{bmatrix}z\end{bmatrix}$,
  $D=\begin{bmatrix}1\end{bmatrix}$,
  $\eps_{\rel,T}=0$, $\eps_{\rel,D}=0$, $\eps_\abs=0$
\item[Parallel composition] block matrices and vectors:\\
  $\eps_{\rel,T}=
    \begin{bmatrix}\eps_{\rel,T}^1 & 0 \\ 0 & \eps_{\rel,T}^2\end{bmatrix}$,
  $\eps_{\rel,D}=
    \begin{bmatrix}\eps_{\rel,D}^1 & 0 \\ 0 & \eps_{\rel,D}^2\end{bmatrix}$,
  $\eps_\abs=\begin{bmatrix}\eps_\abs^1 \\ \eps_\abs^2\end{bmatrix}$.
\end{description}

\subsection{Serial composition}
The serial composition of two filters is more involved.
Let $F$ and $G$ be the ideal linear transfer functions of both
filters, and $\tilde{F}$ and $\tilde{G}$ the transfer functions
implemented over floating-point numbers.

We have $\forall I~ N_\infty(F(I) - \tilde{F}(I)) \leq
\eps_\rel^F . N_\infty(I) + \eps_\abs^F$
(\emph{mutatis mutandis} for $G$). We are interested in
$\eps=N_\infty(F(I) - \tilde{F}(I))$:
that is, a vector of positive numbers indexed by the outputs of the
system such that on every coordinate
$k$, the difference $\delta$ between output $k$ computed over the reals and the
floating-point numbers over the same input $I$ verifies
$\|\delta\|_\infty \leq \eps_k$. We extend $\leq$ to real vectors
coordinate-wise.

The following is easier to understand when each filter has a single input and
a single output; then, all vectors and matrices are scalars (either in
$\bbR$ or $\ratringr$, and
$N_x(v)$ is simply $\|v\|_x$.

The vector $R$ of (re)initialization values is split between $R^F$
(those concerning $F$) and $R^G$ (those concerning $G$).
We split the overall output error of the system between the part that was
introduced by the first filter (and then amplified or attenuated by
the second filter) and the part that was introduced by the second
filter, and use the triangle inequality:
{\footnotesize
\begin{multline}
N_\infty ((G \circ F) (I, R) - (\tilde{G} \circ \tilde{F})(I, R))\\
  \leq N_\infty(G \circ F (I) - G \circ \tilde{F} (I))
     + N_\infty(G \circ \tilde{F} (I) - \tilde{G} \circ \tilde{F}(I))\\
  \leq N_1(G).(F(I) - \tilde{F}(I))
     + \eps_{\rel,T}^G. N_\infty(\tilde{F}(I))
     + \eps_{\rel,D}^G. N_\infty(R^G)
     + \eps_\abs^G\\
  \leq N_1(G).(F(I) - \tilde{F}(I))
     + \eps_{\rel,T}^G.(N_\infty(\tilde{F}(I))+N_\infty(\tilde{F}(I)-F(I)))
     + \eps_{\rel,D}^G. N_\infty(R^G)
     + \eps_\abs^G\\
  \leq (N_1(G) + \eps_{\rel,T}^G).N_\infty(\tilde{F}(I)-F(I))
     + \eps_{\rel,D}^G. N_\infty(R^G)
     +  \eps_\rel^G. N_\infty(F(I)) + \eps_\abs^G\\
  \leq (N_1(G)+\eps_{\rel,T}^G).
       (\eps_{\rel,T}^F.N_\infty(I)+\eps_{\rel,D}^F.N_\infty(R^F)+\eps_\abs^F)\\
     + \eps_{\rel,T}^G. N_1(F). N_\infty(I)
     + \eps_{\rel,D}^G. N_\infty(R^G)
     + \eps_\abs^G\\
  \leq \left[ (N_1(G)+\eps_{\rel,T}^G). \eps_{\rel,T}^F
            + \eps_{\rel,T}^G . N_1(F) \right] . N_\infty(I)\\
     + \left[ (N_1(G)+\eps_\rel^G). \eps_{\rel,D}^F \right]. N_\infty(R^F)
     + \left[ \eps_{\rel,D}^G \right]. N_\infty(R^G)\\
     + \left[ (N_1(G)+\eps_\rel^G). \eps_\abs^F + \eps_\abs^G \right]
\end{multline}}
Thus $\eps_{\rel,T}^{G \circ F}= (N_1(G)+\eps_\rel^G). \eps_\rel^F
            \allowbreak + \allowbreak \eps_\rel^G . N_1(F)$,\\
$\eps_{\rel,D}^{G \circ F}= \begin{bmatrix}
  (N_1(G)+\eps_\rel^G). \eps_{\rel,D}^F & \eps_{\rel,D}^G
  \end{bmatrix}$,
and $\eps_\abs^{G \circ F}=
  (N_1(G)+\eps_\rel^G). \eps_\abs^F + \eps_\abs^G$.

\subsection{Feedback loops}
\label{part:inexact_feedback}
Let us call $o^{(n)}$ the vector of outputs of the filter at step $n$.
It is, ideally, a linear function of the current input, the preceding inputs,
and the preceding outputs.
$O_n = L(I_{\leq n}, O_{< N})$. Let us call $\tilde{L}$ the associated
floating-point function and $\tilde{O}$ the floating-point output of
the filter. Let us call $\Delta = \tilde{O}-O$.
\begin{multline}
\Delta_n
  = \tilde{L}(I_{\leq n}, \tilde{O}_{< N}) - L(I_{\leq n}, O_{< N})\\
  = \tilde{L}(I_{\leq n},\tilde{O}_{< N})-L(I_{\leq n},\tilde{O}_{<N}))
  +  L(I_{\leq n},\tilde{O}_{<N})-L(I_{\leq n}, O_{< N})\\
  = \left(
    \tilde{L}(I_{\leq n},\tilde{O}_{< N})-L(I_{\leq n},\tilde{O}_{<N}))
    \right) + L(0, \Delta_{<N})
\end{multline}
Let
$C_n = \tilde{L}(I_{\leq n},\tilde{O}_{< N})-L(I_{\leq n},\tilde{O}_{<N}))$
be the sequence of vectors of ``error creations'' at each iteration.
Then $\Delta$ verifies the equation $\Delta=C+z T^F_O.\Delta$.
As before, this means $\Delta=(\Id{n}-z T^F_O)^{-1}.C$
and thus that
$N_\infty(\Delta_{\leq n}) \leq
  N_1\left((\Id{n}-z T^F_O)^{-1}\right). N_\infty(C_{\leq n})$.

Let us split $\eps_{\rel,T}^F \in \matrices_{n,n+m}(\bbR_+)$ into
$\eps_{\rel,I}^F \in \matrices_{n,m}(\bbR_+)$ and
$\eps_{\rel,O}^F \in \matrices_{n,n}(\bbR_+)$. Then
\begin{multline}
N_\infty(C_{\leq n}) \leq \eps_{\rel,I}^F.N_\infty(I_{\leq N})
  + \eps_{\rel,O}^F.N_\infty(\tilde{O}_{< N})
  + \eps_{\rel,D}^F.N_\infty(R) + \eps_\abs^F\\
\leq \eps_{\rel,I}^F.N_\infty(I_{\leq N})
  + \eps_{\rel,O}^F.N_\infty(O_{< N})
  + \eps_{\rel,O}^F.N_\infty(
   \underbrace{\tilde{O}_N-O_{< N}}_{\Delta_{< N}})\\
  + \eps_{\rel,D}^F.N_\infty(R) + \eps_\abs^F
\end{multline}

But then, noting $A = N_1\left((\Id{n}-z T^F_O)^{-1}\right)$,
\begin{multline}
N_\infty(\Delta_{\leq n}) \leq
   A. (\eps_{\rel,I}^F.N_\infty(I_{\leq N})
  + \eps_{\rel,O}^F.N_1(T).N_\infty(I_{\leq N})\\
  + \eps_{\rel,O}^F.N_\infty(\Delta_{< N})
  + \eps_{\rel,D}^F.N_\infty(R)
  + \eps_\abs^F)
% a la relecture je ne vois pas d'ou peut sortir le terme suivant
%  + N_1(T^F_0).N_\infty(\Delta_{< N})
\end{multline}
Let $K_1 = A.\eps_{\rel,O}^F \in \matrices_{n,n}(\bbR_+)$
and \begin{equation}
K_2(\iota,\rho) =
  A.\left(\eps_{\rel,I}^F + \eps_{\rel,O}^F.N_1(T)). \iota
  + \eps_{\rel,D}^F.\rho
  + \eps_\abs^F \right.
\end{equation}
Then $N_\infty(\Delta_{\leq n}) \leq
  K_1. N_\infty(\Delta_{< n}) + K_2(N_\infty(I),N_\infty(R))$.
This means that the sequence $u_n=N_\infty(\Delta_{< n})$ verifies
$u_0=0$ and $u_{n+1} \leq K_1.u_n + K_2(N_\infty(I),N_\infty(R))$.
This implies that for all $n$, $u_n$ is less than the least fixed
point $L$ of $v \mapsto K_1.v+K_2(N_\infty(I),N_\infty(R))$.

Recall that the spectral radius of a matrix $M$ of real numbers is the
greatest absolute values of its eigenvalues.
If $K_1$ is contracting (spectral radius less than 1), then
$v \mapsto K_1.v+K_2(N_\infty(I),N_\infty(R))$
has a unique fixed point, by Banach's fixed
point theorem; and $1-K_1$ is invertible.
This fixed point is $v = (1-K_1)^{-1} K_2(N_\infty(I),N_\infty(R))$.
Let
$\eps_{\rel,T} =(1-K_1)^{-1}.A.
   \left(\eps_{\rel,I}^F + \eps_{\rel,O}^F.N_1(T)\right)$,
$\eps_{\rel,D} =(1-K_1)^{-1}.\eps_{\rel,D}^F$,
and $\eps_\abs =(1-K_1)^{-1}.A.\eps_\abs^F$.
Then
$N_\infty(\Delta) \leq \eps_{\rel,T}.N_\infty(I) +
  \eps_{\rel,D}.N_\infty(R) + \eps_\abs$.

Recall that
$K_1 = A.\eps_{\rel,O}^F \allowbreak \in \allowbreak
  \matrices_{n,n}(\bbR_+)$ where
$A$ is the matrix of norms
$N_1\left((\Id{n}-z T^F_O)^{-1}\right)$; $K_1$ bounds the amount of
floating-point imprecision that feeds back into the system. $A$ is the
amplification bounding matrix of the filter consisting merely of the
feedback loop of the original filter; if the original filter is stable
and well-designed, the coefficients of $A$ should be moderate. 
$\eps_{\rel,O}^F$ measures the creation of imprecision in one
iteration of the internal filter; if the filter is numerically
well-designed, then its coefficients are very small. On real-world
examples, $K_1$ was on the order of magnitude of $10^{-15}$.

This suggests an effective method for bounding from above
the various quantities of the form $(1-K_1)^{-1}.y$ that we listed,
where $y$ is a column vector (if $y$ is a matrix, then split it
into its column vectors).
\begin{equation}
d_\infty = (1-K_1)^{-1}.y = \sum_{k=0}^\infty K_1^k.y
\end{equation}
is the unique fixpoint of $\phi = x \mapsto K_1.x+y$, which
is monotonic and contracting.
Consider the matrix norm subordinate to
$\| \cdot \|_\infty$ on vectors:
\begin{equation}
\label{eqn:k1_norm}
\| K_1 \| = \sup_i \sum_j {k_1}_{i,j}
\end{equation}
This gives a rough bound on $d_\infty$:
\begin{equation}
 \| d_\infty \|_\infty \leq \sum_{k=0}^\infty \|K_1\|^k.\|y\|_\infty
  = \frac{\|y\|_\infty}{1-\|K_1\|}.
\end{equation}
Let $d_n = (x \mapsto K_1.x+y)^n(y) = \sum_{k=0}^n K_1^n.y$.
$d_\infty - d_n = K_1^{n+1}.d_\infty$,
thus
\begin{equation}
\|d_\infty - d_n\|_\infty \leq \frac{K_1^{n+1}}{1-\|K_1\|}. \|y\|_\infty.
\end{equation}
Therefore, the following is an upper bound on $d_\infty$:
\begin{equation}\label{eqn:k1_bound}
B = d_n + \left(-\frac{K_1^{n+1}}{\|K_1\|-1}. \|y\|_\infty\right).V_1
\end{equation}
where $V_1$ is a vector of ones of the same dimension as~$y$.
This computation may be effectively performed in floating-point
arithmetic in order to yield a sound upper bound by computing
Eqn.~\ref{eqn:k1_norm} and \ref{eqn:k1_bound} in round-to-$+\infty$
mode ($x \mapsto -1/x$ is monotonic).
Remark that we can directly prove the soundness of the resulting
$\tilde{B}$ by checking that
$K_1.\tilde{B}+y$ is less than $\tilde{B}$ (this checking phase,
though unnecessary assuming a sound implementation, may be cheaply
performed for the sake of security; while it is possible that the
result should be correct and the check fails, this seems very unlikely
in practice, and can be worked around by choosing a slightly larger
$\tilde{B}$).

%% Let $\tilde{\phi}$ be the implementation of $\phi$ in floating-point
%% arithmetic in round to $+\infinity$ node; clearly
%% $\tilde{\phi} \geq \phi$. If we can exhibit $s$ such that
%% $\tilde{\phi}(s) \leq s$, then $s \geq d_\infty$.

\subsection{Trading some accuracy for computation speed; nonlinear elements}
\label{part:trading_accuracy_for_speed}
\begin{figure}
\begin{center}
\input{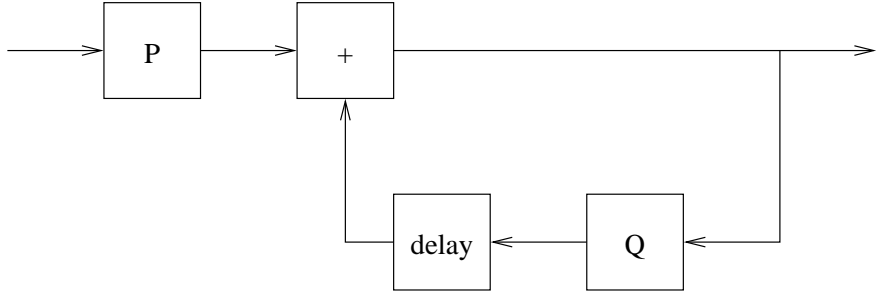}
\caption{An ideal filter equivalent to a filter of Z-transform
$P(z)/(1-Q(z))$.}
\label{fig:ideal_equivalent}
\end{center}
\end{figure}
We have split the behavior of the filter into the sum of the convolution of
the input signal by the power development of a rational function,
representing the exact behavior, and some error term. If we compute
the rational functions exactly over $\ratring{\bbQ}$, then the
rational coefficients might grow expensively large. It seems silly
to use high precision for the coefficients of a system parameterized
by floating-point numbers and implemented with floating-point
errors. Indeed, we may reduce the precision of the coefficients of the
rational function at the expense of adding to the margin of error.

An ideal filter of Z-transform the rational function $P(z)/(1-Q(z))$ where
$P(z)=\sum_{k=0}^{d_p} p_k z^k$ and
$Q(z)=\sum_{k=1}^{d_q} q_k z^k$ with non initialization condition
is equivalent to a filter 
with ideal input Z-transform $P$ and ideal feedback Z-transform $Q$
(Fig~\ref{fig:ideal_equivalent}).
Such a filter may be soundly approximated by a non-ideal
feedback filter $F^\sharp$ with $T^{F^\sharp}_I = P^\sharp$,
$T^{F^\sharp}_O = Q^\sharp$, $\eps_{\rel,I}=\| P^\sharp - P\|_1$,
$\eps_{\rel,I}=\| Q^\sharp - Q\|_1$, $\eps_\abs=0$, which we know how
to solve from Sect.~\ref{part:inexact_feedback}.

More generally: a filter $F$ may be approximated by a filter
$F^\sharp$ with transfer function $T^{F^\sharp} = T^G$,
$\eps^{F^\sharp}_{\rel,T} = \eps^F_{\rel,T} + \eps^G_{\rel,T}$,
$\eps^{F^\sharp}_{\rel,D} = \eps^F_{\rel,D} + \eps^G_{\rel,D}$,
$\eps^{F^\sharp}_\abs = \eps^F_\abs$ where
$G$ is the feedback filter with internal filter $H$ given
$T^H_I = P^\sharp$,
$T^H_O = Q^\sharp$, $\eps^H_{\rel,I}=\| P^\sharp - P\|_1$,
$\eps^H_{\rel,I}=\| Q^\sharp - Q\|_1$, $\eps^H_\abs=0$.

Note that this gives a generic method for approximating non-linear
elements occuring in filters, provided that it is possible to split
them into a linear part and a nonlinear part, the output of which can
be bounded by an affine function of bounds on the absolute value of
the inputs.

\section{Numerical considerations}
\label{part:numerical}
We have so far given many mathematical formulas that are exact in the
\emph{real} field. In this section, we explain how to obtain sound
abstractions for these formulas using floating-point arithmetics.

\subsection{Interval arithmetics}
IEEE floating-point arithmetics \cite{IEEE754} and good extended
precision libraries such as MPFR \cite{GMP} provide functions
computing \emph{upward rounded} (or \emph{rounded-to-$+\infty$})
and  \emph{downward rounded} (or \emph{rounded-to-$-\infty$})
results: that is, if $f(x_1,\ldots,x_n)$ is the exact operation on
real numbers and $\tilde{f}^-$ and $\tilde{f}^+$
are the associated floating-point downward and upward operations,
then $f(x_1, \ldots, x_n)$ is guaranteed to be in the interval
$[\tilde{f}^-(x_1, \ldots, x_n), \tilde{f}^+(x_1,\ldots,x_n)]$,
which will guarantee the \emph{soundness} of our approach.
Furthermore, for many operations, 
$\tilde{f}^-(x_1, \ldots, x_n)$ and $\tilde{f}^+(x_1,\ldots,x_n)$
are guaranteed to be optimal; that is, no better bounds can be
provided within the desired floating-point format; this will guarantee
\emph{optimality} of certain of our elementary operations.

\subsection{Approximate algebraic computations}
\label{part:approximate_algebraic_structures}
In many occasions, we ideally would like to compute on real
polynomials $P=\sum_{k=1}^n p_k z^k$ but instead we compute on
floating-point polynomials $\tilde{P}=\sum_{k=1}^n [l_k,h_k] z^k$
abstracting the set $\gamma(\tilde{P})$ of polynomials $P$ such that
$\forall k~p_k \in [l_k,h_k]$. In practice, it will often be necessary
that $0 \notin [l_k,h_k]$ in order to avoid uncertainties on the
degree of the polynomial. All the usual polynomial operations
(addition, multiplication by a scalar, subtraction, multiplication)
may be abstracted using interval arithmetics. We also include a test
$\textrm{contains}_0(\tilde{P})$
whether the null polynomial is in $\gamma(\tilde{P})$.
We call this structure an \emph{abstract ring}.

Given a abstract ring $R$, we construct the \emph{abstract field} of
fractions over that ring using the following operations:
$p_1/q_1 + p_2/q_2 \allowbreak = \allowbreak (p_1 q_2 + p_2 q_1) / (q_1 q_2)$,
$k.(p / q) = (k . p) / q$,
$(p_1 / q_1) . (p_2 / q_2) \allowbreak = \allowbreak
  (p_1.p_2) \allowbreak / \allowbreak (q_1.q_2)$,
$(p_1 / q_1) / (p_2 / q_2) \allowbreak = \allowbreak
  (p_1.q_2) \allowbreak / \allowbreak (q_1.p_2)$,
$\textrm{contains}_0(p/q) = \textrm{contains}_0(p)$.
We can make a simple attempt at reducing the fractions by checking
that there are no trivial cancellations between the numerator and
denominator in products and quotients.

Given an abstract ring $K$, we construct the \emph{abstract ring of
matrices} over that ring with the usual operations:
if $M=A+B$, $m_{i,j}=a_{i,j}+b_{i,j}$; if $M=A.B$,
$m_{i,j} = \sum_k a_{i,k}.b_{k,j}$. If $K$ is an abstract field, we
can also implement Gaussian elimination in order to compute
$A^{-1}.B$ given a square matrix $A$ and a matrix $B$. When we look
for a pivot, we select elements $e$
such that $\textrm{contains}_0(e)$ is false.

Unfortunately, computations on such approximate structures may yield
unfavorable results. In particular, the absence of simplification
between the numerator and denominator may yield fractions
$\tilde{P}(z)/\tilde{Q}(z)$ where $\tilde{P}$ and $\tilde{Q}$ have
some common zeroes. The spurious poles that are introduced  not be
that much of a problem if we use partial fraction decomposition
(Sect.~\ref{part:partial_fraction_decomposition}), for they will yield
very small coefficients in the decomposition; however, they will make
the computations more complex. If using the simple tail bounds of
Sect.~\ref{part:inverse_product_norm}, the results may be considerably
worse.

A solution is to perform all computations on rational functions
exactly over $\ratring{\bbQ}$. Then, cancellation between a numerator
and a denominator can be performed exactly by division by their
greatest common divisor, which is obtained from Euclide's algorithm
over the Euclidean division of polynomials. No spurious poles may be
introduced. However, on large filter networks, exact computations may
produce exceedingly large integer numerators and denominators. It is
then possible to apply the approximation scheme of
Sect.~\ref{part:trading_accuracy_for_speed} in order to trade speed
for potential precision. This is the solution that we implemented in
our system: exact computations on rational numbers and safe
approximations to limit the length of the numbers involved in the
computations.

\subsection{Computation of developments}
\label{part:compute_initial_norm1}
When bounding the norm $\| P/Q \|_1$ of a series quotient of two polynomials,
we split the series into its $N$ initial terms of development, which we compute
explicitly, and a tail whose norm we bound.
The first idea is to compute the $N$ first terms of the series by
quotienting the series, as explained in
Sect.~\ref{part:formal_power_series} or, equivalently, by running the
filter for $N$ iterations on the Dirac input $1,0,0,\ldots$. In order
to provide a sound result, one would work using interval arithmetics
over floating-point numbers.
However, as already noted by Feret, after some number of iterations
the sign of the terms becomes unknown and then the magnitude of the
terms increase fast; it is therefore indicated to compute the
development until the first term of unknown sign is reached, and
assign $N$ accordingly (one may still also enforce a maximal number of
iterations $N_{\max}$).
In order to be able to develop the quotient further with good
precision, one can use a library of extended-precision floating-point
computations with selectable rounding direction, such as the MPFR
library now part of GNU~MP \cite{GMP}.

%% Another occasion in which we have to compute the norm of the initial
%% development of a fraction is the scheme described in
%% Sect.~\ref{part:inverse_product_norm}, where we develop
%% $1/\prod_k (z-|\xi_k|)$ where the $\xi$ are the roots of a polynomial,
%% known with some bounded imprecision (so $|xi_k|$ is known only as
%% belonging to some interval).
%% The phenomenon explained in the preceding paragraph --- after a number
%% of iterations, the sign of the coefficients becomes unknown and the
%% magnitude of the errors explode --- also
%% occurs here, which is all the more distressing since the exact
%% coefficients are known to be positive!
%% Because in this case the coefficients of
%% the $Q$ denominator polynomial are themselves known to some limited
%% precision, the usefulness of extended precision computations is
%% limited --- one would have to obtain high precision on the $\xi$,
%% which may be difficult (most available libraries for obtaining roots
%% are implemented using IEEE floating-point arithmetics, not
%% extended-precision libraries) and costly.
%% On the other hand, here we known the \emph{exact form} of the
%% coefficients: they are obtained by the convolution of the developments
%% of $1/(z-|\xi_k|)$, that is, $\sum_i |\xi|^{-(i+1)} z^i$. Given a
%% (small) interval for $|\xi|$, it is possible to bound
%% $|\xi|^{-(i+1)}$ with high precision. The bounds on
%% $\|1/\prod_k (z-|\xi_k|)\|_1^{<N}$ thus obtained are considerably more
%% precise than those obtained by development of the quotient.

\subsection{Bounding the roots}
In order to bound $\| P/Q \|_1$, where $P$ and $Q$ may possibly be
given using interval coefficients, we have to bound the roots of $Q$.
More formally, we have to solve the following problem:
given an interval polynomial
$\tilde{P}(z) = \sum_{k=1}^n [l_k; h_k] z^k$ such that
$0 \notin [l_n,h_n]$,
find a family $(\tilde{\xi}_k,\rho_k)_{1 \leq k \leq n}$
($\xi_k \in \bbC$ with $\Re \xi_k$ and $\Im \xi_k$ floating-point
numbers, $\rho_k \in \bbR_+$ a floating-point number) such that for
any polynomial $P = \sum_{k=1}^n p_k$ such that
$\forall k~ p_k \in [l_k,h_k]$,
then, up to a permutation, the $n$ roots $(\xi_k)_{1 \leq k \leq n}$
of $P$ are such that $\xi_k \in D(\tilde{\xi}_k,\rho_k)$ where
$D(z,r)$ is the closed disc of center $z$ and radius~$r$.

Often, what we need is actually bounds on the $|\xi_k|$; this can
easily be obtained from the preceding bounds using interval arithmetic
on plus, minus, multiply and square root.

Our coefficients are intervals $[l_k,h_k]$ in order to accommodate
possible errors of floating-point computations. As a consequence, it
is expected that $h_k-l_k$ are small. This suggests to us a two-step
method for obtaining the desired bounds:
\begin{enumerate}
\item Use an efficient and, in practice, very accurate algorithm to
  obtain \emph{approximations} $x_j$ to the roots of
  $\sum_{k=1}^n \frac{l_k+h_k}{2} z^k$ (the ``midpoint polynomial'').
\item From those approximations, obtain bounds on the radius of the
  error committed.
\end{enumerate}

There exist a variety of methods and implementations to perform the
first point. We used \texttt{gsl\_poly\_complex\_solve} of the
GNU~Scientific Library \cite{GSL}, which is based on an eigenvalue
decomposition of the companion matrix.

For the second step, Rump describes a variety of bounding methods
\cite{Rump_JCADP03} which take a polynomial and approximate roots as
an input and output error radii; these methods may be performed
using interval arithmetics. We implemented the simplest and roughest one:
$\xi_j$ is in a closed disc of center $x_j-p_j$ and radius $|p_j|$
where
\begin{equation}
p_j = \frac{n P(x_j)}{p_n \prod_{k \neq j} x_j - x_k},
\end{equation}
which is easily implemented using interval arithmetics ($P$ becomes
$\tilde{P}$ etc.).

\section{Implementation and case studies}
\label{part:experiments}
We implemented the algorithms described above in a simple Objective
Caml \cite{OCaml} program: filters are represented by a record of all their
characteristics (transfer matrices, bounding matrices); functions (in
the OCaml) sense construct filter records, or perform composition
operations.

The formal computations on fractions are performed over $\bbQ$,
implemented using GNU~MP's \texttt{mpq} type \cite{GMP}. We initially
considered using MPFR \cite{MPFR}, an extended precision library with
sound rounding modes, for interval computations; instead, we simply
use the IEEE-754 rounding modes of the hardware floating-point unit,
which is much faster.

\subsection{Composition of TF2 filters}
Let us recall the example of Sect.~\ref{part:reconstitution}. It is a
composition of two TF2 filters with a feedback loop around it.
The serial composition of the filter in
Fig.~\ref{fig:compound_filter_example1} and
another TF2 filter, all with realistic coefficients,
is analyzed in about 0.10~s on a recent PC;
the analyzer finds that $\| S \| \leq g \| E \|$ with
$g \simeq 2$, with $\eps_\rel \simeq 10^{-12}$
and $\eps_\abs \simeq 10^{-305}$.

The power series developments of rational functions
(Sect.~\ref{part:compute_initial_norm1}) are done up to around
the 27th order.

\subsection{Complex nonlinear iterated filter}

We now consider a nonlinear, iterated filter due to Roozbehani et
al. \cite{RoozbehaniEtAl_HSCC05}[§5]. We first analyze separately
\texttt{filter1()} (2nd-order linear filter) and
\texttt{filter2()} (2nd-order affine filter).
So as to simplify matters, we do not give the transfer functions using
matrices, matrices inverses etc. but as the solution of a system of
linear equations over polynomials in $z$. We obtain that system very
simply from the program: whenever we see an assignment $x := e$, we
turn it into an equation $x = e$ (we assume without loss of
generalities that variables are only assigned once in a single
iteration step), where $e$ is the original expression where a variable $v$
that has not yet been assigned in the current iteration is
replaced by $i_v + z.v$, $i_v$ standing for the initialization value
of $v$.

{\small\tt
\begin{tabbing}
void filter1 () \{\\
\ \ static float E[2], S[2];\\
\ \ if (INIT1) \{\\
\ \ \ \ S[0] = X; P = X;\\
\ \ \ \ E[0] = X; E[1]=0; S[1]=0;\\
\ \ \} else \{\\
\ \ \ \  P =0.5*X-0.7*E[0] +0.4*E[1] \= $p=0.5e -0.7(i_{e_0}+z.e_0)$\\
\ \ \ \ \ \ +1.5*S[0]-S[1]*0.7; \>\quad $+0.4(i_{e_1}+z.e_1)$\\
\>\quad $+1.5(i_{s_0}+z.s_0) -0.7(i_{s_1}+z.s_1)$\\
\ \ \ \ E[1] = E[0];\> $e_1=i_{e_0}+z.e_0$\\
\ \ \ \ E[0] = X; \> $e_0=e$\\
\ \ \ \ S[1] = S[0];\> $s_1=i_{s_1}+z.e_1$\\
\ \ \ \ S[0] = P;\> $s_0 = p$\\
\ \ \ \ X=P/6+S[1]/5;\> $x=p/6+s_1/5$\\
\ \ \}\\
\}
\end{tabbing}}
We call $e$ the input value for \texttt{X}.
We solve the system and obtain
$x=Q.e+Q_{i_{e_0}}.i_{e_0}+Q_{i_{e_1}}.i_{e_1}
      +Q_{i_{s_0}}.i_{s_0}+Q_{i_{s_1}}.i_{s_1}$.
The common denominator of the $Q$ fractions is $10-15z+7z^2$, which
has complex conjugate roots $z$ such that $|z| \simeq 1.2$.
$i_{e_1}=i_{s_1}=0$ and $i_{e_0}=i_{s_0}=\iota$ (the last value
for input $e$ such that \texttt{INIT1} is true), thus
$\| x \|_\infty \leq \|Q\|_1 . \|e\|_\infty+\| Q_{i_{e_0}}+Q_{i_{s_0}}\|_\infty.\|\iota\|$.
With a precondition $\| e\|_\infty \leq 400$, this yields $\| x \|_\infty <
339$. If we take the coarser inequality
$\| x \|_\infty \leq \|Q\|_1 . \|e\|_\infty+
                (\| Q_{i_{e_0}}\|_\infty +\|Q_{i_{s_0}}\|_\infty).\|\iota\|$ we
get $\| x \|_\infty < 528$. Roozbehani et al. find a bound $\simeq 531$.

{\small\tt
\begin{tabbing}
void filter2 () \{\\
\ \ static float E2[2], S2[2];\\
\ \ if (INIT2) \{\\
\ \ \ \ S2[0] =0.5*X; P = X;\\
\ \ \ \ E2[0] = 0.8*X; E2[1]=0; S2[1]=0;\\
\ \ \} else \{\\
\ \ \ \ P =0.3*X-E2[0]*0.2+E2[1]*1.4 \= $p=0.3e-0.2(i_{e_0}+z.e_0)$\\
\ \ \ \ \ \ +S2[0]*0.5-S2[1]*1.7; \> \quad $+1.4(i_{e_1}+z.e_1)$\\
\>\quad $+0.5(i_{s_0}+z.s_0)+1.7(i_{s_1}+z.s_1)$\\
\ \ \ \ E2[1] = 0.5*E2[0];\> $e_1=0.5(i_{e_0}+z.e_0)$\\
\ \ \ \ E2[0] = 2*X; \> $e_0=2e$\\
\ \ \ \ S2[1] = S2[0]+10; \> $s_1=i_{s_0}+z.s_0+\tau$\\
\ \ \ \ S2[0] = P/2+S2[1]/3; \> $s_0=p/2+s_1/3$\\
\ \ \ \ X=P/8+S2[1]/10; \> $x=p/8+s_1/10$\\
\ \ \}\\
\}
\end{tabbing}}
We proceed similarily (with the introduction of $\tau=10/(1-z)$ and obtain
$x=Q.e+Q_{i_{e_0}}.i_{e_0}+Q_{i_{e_1}}.i_{e_1}
      +Q_{i_{s_0}}.i_{s_0}+Q_{i_{s_1}}.i_{s_1} + Q_c$.
The common denominator of the $Q$ is $60+35z+51 z^2$, with complex
conjugate roots $z$ such that $|z| \simeq 1.08$.
Then
$\| x \|_\infty \leq \|Q\|_1.\|e\|_\infty+\| 0.8 Q_{i_{e_0}}+0.5
                Q_{i_{s_0}}\|_\infty.\|\iota\| + \| Q_c\|_\infty$.
This yields $\| x \|_\infty \leq 1105$. 

The two linear filters are combined into an iterated nonlinear filter.
\texttt{filter1()} (resp. \texttt{filter2()}) is run with a
pre-condition of $\texttt{X} \in [-400,400]$
(resp. $[-800,800]$). We replace the call to the filter by its postcondition
$\texttt{X} \in [-339,339]$ (resp. $\texttt{X} \in [-1105,1105]$).

{\small\noindent\tt%
void main () \{\\
\hbox{~~}X = 0;\\
\hbox{~~}INIT1 = TRUE; INIT2=TRUE;\\
\hbox{~~}while (TRUE) \{\\
\hbox{~~~~}X = 0.98 * X + 85;\\
\hbox{~~~~}if (abs(X)<= 400) \{\\
\hbox{~~~~~~}filter1 ();\\
\hbox{~~~~~~}X=X+100;\\
\hbox{~~~~~~}INIT1=FALSE;\\
\hbox{~~~~}\} else\\
\hbox{~~~~}if (abs(X)<=800) \{\\
\hbox{~~~~~~}filter2();\\
\hbox{~~~~~~}X=X-50;\\
\hbox{~~~~~~}INIT2=FALSE;\\
\hbox{~~~~}\}\\
\ \ \}\}
}

The program then can be abstracted into:\\
{\tt while (TRUE) \{\\
\hbox{~~}X = 0.98 * X + 85;\\
\hbox{~~}maybe choose X in $[-1155, 1055]$;\\
\}}

We obtain $\texttt{X} \in [-1155,4250.02]$ by running
Astr\'ee with a large number of narrowing iterations, whereas Astr\'ee
cannot analyze the original program precisely and cannot bound
\texttt{X}.
In this case, the exact solution $[-1155,4250]$ ($x=0.98x+85$ has for
unique solution $x=4250$) could have been computed algebraically, but
in more complex filters this would not have been the case. Roozbehani
et al. have a bound of $4560$.

Note that the non-abstracted
program converges to a value $\simeq 205$, with $\texttt{X} \in [0,
209]$. However, this very simple program illustrates our methodology
for compositional analysis: finding the optimal solution is possible
here because the program is simple, but would not be possible in
practice  if we had added more nonlinear behavior and nondeterministic
inputs, as in real-life reactive code; whereas by analyzing precisely
each linear filter and plugging the results back into a generic
analyzer, we get reasonable results.

\section{Related works}
In the field of digital signal processing, some sizable literature
has been devoted to the study of the effects of fixed-point and
floating-point errors on numerical filters. In the area of fixed-point
computation, bounds on the sizes of the various operands are of
paramount importance: operands that leave the prescribed range will
undergo saturation and the output signal will be distorted. For these
reasons, operands are scaled so as not to produce digital saturation;
yet, the scale factor should be made large enough that rounding errors
are very small compared to the typical magnitude of the
signal. While the fact that the l1-norm of the convolution kernel is
what matters for judging overflow, it is argued that this norm is
``overly pessimistic'' \cite[§11.3]{Jackson_Digital_filters}
\cite[eq~13]{Jackson70_Bell}, not to mention the difficulties in
estimating it. In practice, filter designers have
preferred criteria that indicate no saturation for most
``commonplace'' inputs, excluding pathological inputs. Our vision is
different: our results must be sound in all circumstances, even
pathological inputs.

The impact of fixed- and floating-point errors in digital filters was
classically studied from by modeling the errors as random sources of
known distribution, independent of each other and with no
temporal correlation (i.e. correlations between successive values)
\cite{BomarElAl_IEEETransCircuits97,Rao_IEEETransSignal92}.
These assumptions are, in reality, false: the computational process is
fully deterministic, and not random; the computations are generally
interdependent (all computations inside a filter depend on the past of
the input variables); and there are temporal correlations.
However, circuit designers are concerned with the spectral
distribution of output noise \cite{Jackson70_Bell},
and optimization of hardware or software
implementations with respect to this noise, and these tools are
adequate for this. On the other hand, we merely aim at providing sound
bounds for the outputs of the system, but the bounds that we provide
must be sound without any extra and unfounded suppositions.

J.~Feret has proposed an abstract domain for analyzing programs
comprising digital linear filters \cite{Feret_ESOP04}. He provides
effective bounds for first and second degree filters. In comparison,
we consider more complex filter networks, in a compositional fashion;
but we analyze specifications, and not C code (which is usually
compiled from those specifications, with considerable loss of
structure). Another difference is that we do not perform abstract
iterations.
Feret's method currently considers only second-order
filters (i.e. TF2), though it may be possible to adapt it to
higher-order filters. On second-order filters, the bounds computed by
Feret's method and the method in this paper are very close (since both
are based on a development of the convolution kernel, though they use
different methods of tail estimation).

Lamb et al. \cite{Lamb_et_al_PLDI03} have proposed effective methods,
based on linear algebra, for computing equivalent filters for DSP
optimization. They do not compute bounds, nor do they study
floating-point errors.

Roozbehani et al. \cite{RoozbehaniEtAl_HSCC05}
find program invariants by Lagrangian relaxation and semidefinite
programming, with quadratic invariants. In order to make problems
tractable, they too apply a blockwise abstraction. The class of
programs that they may analyze directly is potentially larger, but the
results are less precise than our method on some linear filters. They
do not handle floating-point imprecisions (though this can perhaps be
added to their framework).

One possible application of our method would be to integrate it as a
pre-analysis pass of a tool such as Astr\'ee
\cite{ASTREE_ESOP05}. Astr\'ee computes bounds on all floating-point variables
inside the analyzed program, in order to prove the absence of errors
such as overflow. In order to do so, it needs to compute reasonably
accurate bounds on the behavior of linear filters. A typical
fly-by-wire controller contains dozens of TF2 filters, some of which
may be integrated into more complex feedback loops; in some cases,
separate analysis of the filters may yield too coarse bounds.

\section{Conclusions and future works}
We have proposed effective methods for providing sound bounds on the
outcome of complex linear filters from their flow-diagram
specifications, as found in many
applications. Computation times are modest; furthermore, the nature
of the results of the analysis may be used for modular analyses ---
the analysis results of a sub-filter can be stored and never be
recomputed until the sub-filter changes.

The usefulness of these methods is twofold. First, they could
be directly implemented in the graphical user interface for designing
circuits. Users may then be able to compute gains or to check the
stability of filters, taking into account floating-point errors (which
conventional Z-transform techniques do not consider). Second, they can
be used as a way to automatically obtain static analysis
``transformers'' or ``transfer functions'': a static analysis tool
such as Astrée may detect that some program sequence implements such
or such complex linear filter, and apply some invariant relation
computed using the techniques in that paper.

In future works, we will examine the case of non-linear filters and
compositional, modular analysis. The analysis of a combination of linear and
non-linear filters can be done in two ways or a combination thereof:
\begin{itemize}
\item
the overall behavior of a nonlinear filter may be constrained by some
input-output relationship such as $\| O \|_\infty \leq (1+\epsilon) \|
I \|_\infty$ (example of a rate limiter), and this input-output
relationship can be integrated into the abstract semantics as in
Part~\ref{part:float_compositional};
\item
the overall behavior of a linear filter can be precisely bounded, and
this bound information can be fed into an analysis of a larger
nonlinear filter, such as one based on statically computed
relationships between intervals \cite{Monniaux_SAS07}
\end{itemize}

\bibliographystyle{plain}
\bibliography{caml,numerical,signal,absint,David_Monniaux,engineering}

\begin{thebibliography}{10}

\bibitem{BlanchetCousotEtAl02-NJ}
B{.} Blanchet, P{.} Cousot, R{.} Cousot, J{.} Feret, L{.} Mauborgne, A{.}
  Min{\'e}, D{.} Monniaux, and X{.} Rival.
\newblock Design and implementation of a special-purpose static program
  analyzer for safety-critical real-time embedded software.
\newblock In {\em The Essence of Computation: Complexity, Analysis,
  Transformation}, number 2566 in Lecture Notes in Computer Science, pages
  85--108. Springer Verlag, 2002.

\bibitem{BlanchetCousotEtAl_PLDI03}
B{.} Blanchet, P{.} Cousot, R{.} Cousot, J{.} Feret, L{.} Mauborgne, A{.}
  Min{\'e}, D{.} Monniaux, and X{.} Rival.
\newblock A static analyzer for large safety-critical software.
\newblock In {\em PLDI}, pages 196--207. ACM, 2003.

\bibitem{BomarElAl_IEEETransCircuits97}
Bruce~W. Bomar, L.~Montgomery Smith, and Roy~D. Joseph.
\newblock Roundoff noise analysis of state-space digital filters implemented on
  floating-point digital signal processors.
\newblock {\em {IEEE} Trans. on Circuits and Systems II}, 44(11):952--955,
  1997.

\bibitem{LUSTRE}
Paul Caspi, Daniel Pilaud, Nicolas Halbwachs, and John~A. Plaice.
\newblock {LUSTRE}: a declarative language for real-time programming.
\newblock In {\em POPL '87: Proceedings of the 14th ACM SIGACT-SIGPLAN
  symposium on Principles of programming languages}, pages 178--188. ACM Press,
  1987.

\bibitem{ASTREE_ESOP05}
Patrick Cousot, Radhia Cousot, J\'er\^ome Feret, Laurent Mauborgne, Antoine
  Min\'e, David Monniaux, and Xavier Rival.
\newblock The {ASTR\'EE} analyzer.
\newblock In {\em ESOP}, number 3444 in Lecture Notes in Computer Science,
  pages 21--30, 2005.

\bibitem{Feret_ESOP04}
Jérôme Feret.
\newblock Static analysis of digital filters.
\newblock In {\em ESOP '04}, number 2986 in {L}ecture {N}otes in {C}omputer
  {S}cience. Springer-Verlag, 2004.

\bibitem{GMP}
Free Software Foundation.
\newblock {\em {GMP} --- {GNU} multiple precision arithmetic library}, 2004.

\bibitem{GSL}
Free Software Foundation.
\newblock {\em {GSL} --- {GNU} scientific library}, 2004.

\bibitem{IEEE754}
IEEE.
\newblock {\em Standard for Binary Floating-Point Arithmetic}.
\newblock standard 754.

\bibitem{MPFR}
INRIA et al.
\newblock {\em The MPFR Library}.

\bibitem{Jackson70_Bell}
Leland~B. Jackson.
\newblock On the interaction of roundoff noise and dynamic range in digital
  filters.
\newblock {\em The {B}ell System Technical J.}, 49(2):159--184, February 1970.

\bibitem{Jackson_Digital_filters}
Leland~B. Jackson.
\newblock {\em Digital Filters and Signal Processing}.
\newblock Kluwer, 1989.

\bibitem{MR1646107}
Peter Kirrinnis.
\newblock Partial fraction decompostion in {$\mathbb{C}(z)$} and simultaneous
  {N}ewton iteration for factorization in {$\mathbb{C}[z]$}.
\newblock {\em J. Complexity}, 14(3):378--444, 1998.

\bibitem{Lamb_et_al_PLDI03}
Andrew~A. Lamb, William Thies, and Saman Amarasinghe.
\newblock Linear analysis and optimization of stream programs.
\newblock In {\em PLDI '03}, pages 12--25. ACM, 2003.

\bibitem{OCaml}
Xavier Leroy.
\newblock {\em The Objective Caml system, documentation and user's guide}.
\newblock INRIA.

\bibitem{Ariane501}
Jacques-Louis Lions et~al.
\newblock {A}riane 501: Flight 501 failure.
\newblock Technical report, ESA / CNES, 1996.
\newblock Available on WWW.

\bibitem{Mine:ESOP04}
A{.} Min\'e.
\newblock Relational abstract domains for the detection of floating-point
  run-time errors.
\newblock In {\em ESOP'04}, volume 2986 of {\em LNCS}, pages 3--17. Springer,
  2004.

\bibitem{Monniaux_CAV05}
David Monniaux.
\newblock Compositional analysis of floating-point linear numerical filters.
\newblock In {\em Computer-aided verification: CAV '05}, number 3576 in Lecture
  Notes in Computer Science, pages 199--212. Springer Verlag, 2005.

\bibitem{Monniaux_SAS07}
David Monniaux.
\newblock Optimal abstraction on real-valued programs.
\newblock In {\em Static analysis symposium (SAS)}, 2007.
\newblock To appear.

\bibitem{Rao_IEEETransSignal92}
Bhaskar~D. Rao.
\newblock Floating point arithmetic and digital filters.
\newblock {\em {IEEE} Trans. on Signal Processing}, 40(1):85--95, January 1992.

\bibitem{RoozbehaniEtAl_HSCC05}
M.~Roozbehani, E.~Feron, and A.~Megretski.
\newblock Modeling, optimization and computation for software verification.
\newblock In {\em HSCC}, number 3414 in LNCS, page 606. Springer, 2005.

\bibitem{Rump_JCADP03}
Siegfried~M. Rump.
\newblock Ten methods to bound multiple roots of polynomials.
\newblock {\em J. of Computational and Applied Math.}, 156(2):403--432, 2003.

\end{thebibliography}

\appendix
%\section{Mathematical lemmas}
For any matrix $M$, let us note $\minor_{i,j}(M)$ the
determinant of the matrix obtained by removing line~$i$ and column~$j$
from~$M$.
We recall that for any matrix $M$ of dimension $n$
\begin{equation}
\det(M)=\sum_{j=1}^n (-1)^{n-1} m_{i,j}.\minor_{1,j}(M)
\end{equation}
and that the determinant is $n$-linear.
Recall that for any matrix $M$ of invertible determinant,
\begin{equation}
M^{-1}=\det(M)^{-1}.\begin{bmatrix}\minor_{i,j}(M)\end{bmatrix}^t
\end{equation}
\label{eqn:minors_inversion}

\begin{lem}
If $A \in \matrices_{n,n}(\ratringr)$, then there exists
$B \in \ratringr$ such that
$\det(\Id{n}-zA)=1-zB$.
\end{lem}

\begin{proof}
Proof by induction on $n$. The case $n=1$ is trivial.
Now let us consider $n>1$.
\begin{multline}
\det(\Id{n}-zA)\\
=(1-z a_{1,1}) \minor_{1,1}(\Id{n}-zA)
  + \sum_{j=2}^n (-1)^n z a_{1,j} \minor_{1,j}(\Id{n}-zA)\\
= \minor_{1,1}(\Id{n}-zA)
  + z \sum_{j=1}^n (-1)^n z a_{1,j} \minor_{1,j}(\Id{n}-zA)
\end{multline}
The result follows by the application of the induction hypothesis,
and the fact that $B \ratringr$ is a ring and thus the determinant
of any matrix over that ring is itself in the ring.
\end{proof}

\begin{cor}
If $A \in \matrices_{n,n}(\ratringr)$, then 
$\Id{n}-zA)$ has an inverse in $\matrices_{n,n}(\ratringr)$.
\label{cor:feedback_matrix_inversion}
\end{cor}

\begin{proof}
By the preceding lemma,
$\det(\Id{n}-zA)$ is of the form $1-zP(z)/Q(z)$,
where $P$ and $Q$ are polynomials such that the constant coefficient
of $Q$ is $1$, therefore $(\det(\Id{n}-zA))^{-1} = Q(z)/(Q(z)-zP(z)$
is in $\ratringr$.
All the $\minor_{i,j}(\Id{n}-zA)$ are elements of $\ratringr$,
the result follows by applying Equ.~\ref{eqn:minors_inversion}.
\end{proof}

\end{document}